\documentclass[lettersize,journal]{IEEEtran}
\usepackage{amsmath,amsfonts}
\usepackage{algorithmic}
\usepackage{algorithm}
\usepackage{array}
\usepackage{bm}
\usepackage{textcomp}
\usepackage{stfloats}
\usepackage{url}
\usepackage{verbatim}
\usepackage{slashbox}
\usepackage{graphicx}
\usepackage{cite}
\usepackage{xcolor}
\usepackage[font=small,labelfont=bf]{caption}
\usepackage[font=small,labelfont=bf]{subcaption}
\usepackage{makecell}
\usepackage{multicol}
\usepackage{multirow}
\usepackage{booktabs}
\usepackage{mathtools}
\usepackage{enumitem}
\newtheorem{theorem}{Theorem}

\newenvironment{proof}{ \textbf{Proof:} }{ \hfill $\Box$}

\def\bb0{{\mathbb{0}}}

\def\ba{{\mathbf{a}}}
\def\bb{{\mathbf{b}}}

\def\bh{{\mathbf{h}}}

\def\b0{{\mathbf{0}}}

\def\bF{{\mathbf{F}}}

\def\bI{{\mathbf{I}}}

\def\bS{{\mathbf{S}}}

\def\bW{{\mathbf{W}}}

\def\bbC{{\mathbb{C}}}

\def\bbR{{\mathbb{R}}}

\def\bbZ{{\mathbb{Z}}}

\def\cG{\mathcal{G}}

\def\cI{\mathcal{I}}

\def\cO{\mathcal{O}}

\def\cR{\mathcal{R}}

\def\cU{\mathcal{U}}

\def\sfA{\mathsf{A}}
\def\sfB{\mathsf{B}}

\def\sfD{\mathsf{D}}
\def\sfE{\mathsf{E}}

\def\sfL{\mathsf{L}}

\def\sfN{\mathsf{N}}

\def\sfP{\mathsf{P}}

\def\sfR{\mathsf{R}}

\def\sfZ{\mathsf{Z}}

\def\sfc{{\mathsf{c}}}

\def\sfi{{\mathsf{i}}}
\def\sfj{{\mathsf{j}}}

\def\sfm{{\mathsf{m}}}

\def\sfv{{\mathsf{v}}}

\def\sfx{{\mathsf{x}}}
\def\sfy{{\mathsf{y}}}
\def\sfz{{\mathsf{z}}}
\def\sf0{{\mathsf{0}}}

\def\rmF{\mathrm{F}}

\def\rmT{\mathrm{T}}

\def\rmc{{\mathrm{c}}}

\def\rm0{{\mathrm{0}}}

\def\b0{{\pmb{0}}} 

\def\kron{\otimes}

\DeclarePairedDelimiter\ceil{\lceil}{\rceil}
\DeclarePairedDelimiter\floor{\lfloor}{\rfloor}
\begin{document}
	
	\title{Beamforming with hybrid reconfigurable parasitic antenna arrays
	}

		\author{\IEEEauthorblockN{Nitish Vikas Deshpande, \textit{Student Member, IEEE,} Miguel Rodrigo Castellanos, \textit{Member, IEEE,} Saeed R. Khosravirad, \textit{Senior Member, IEEE,} Jinfeng Du,   \textit{Senior Member, IEEE,}   Harish Viswanathan, \textit{Fellow, IEEE,} and Robert W. Heath Jr., \textit{Fellow, IEEE}}\\ 
			\thanks{ {Nitish  Vikas Deshpande  and Robert W. Heath Jr. are
				with the Department of Electrical and Computer Engineering,  University of California San Diego, La Jolla, CA 92093 (email: \{nideshpande, rwheathjr\}@ucsd.edu).}
			Miguel Rodrigo Castellanos is with the Department of Electrical  Engineering and Computer Science at the University of Tennessee, Knoxville, TN 37996 USA (e-mail: mrcastellanos@utk.edu).
			Saeed R. Khosravirad, Jinfeng Du, and Harish Viswanathan are with Nokia
			Bell Laboratories, Murray Hill, NJ 07974, USA (email: \{saeed.khosravirad,  jinfeng.du,
			harish.viswanathan\}@nokia-bell-labs.com).
			This material is also based upon work supported by the National Science Foundation under grant nos. NSF-ECCS-2435261, NSF-CCF-2435254, NSF-ECCS-2414678, and the Army Research Office under Grant W911NF2410107.
		}
	}

	\maketitle
	
	\begin{abstract}
		A parasitic reconfigurable antenna array is a low-power approach for beamforming  using  passive tunable elements.
		Prior work on  reconfigurable antennas in   communication theory is based on ideal radiation pattern abstractions. It does not
		 address the problem of physical realizability.
		Beamforming with parasitic elements is inherently difficult because  mutual coupling creates non-linearity in the beamforming gain objective. We develop a multi-port circuit-theoretic  model of the hybrid array with parasitic  elements and  antennas with active RF chain validated through electromagnetic simulations with a dipole array. We then derive the beamforming weight of the parasitic element using the  theoretical beam pattern expression for the case of a single active antenna and multiple parasitic elements. We show that the parasitic beamforming is challenging because the weights are subject to coupled magnitude and phase constraints. We simplify the beamforming  optimization problem using a shift-of-origin transformation to the typical unit-modulus beamforming weight.
	With this transformation, we derive a closed-form solution for the reconfigurable parasitic reactance. We generalize this solution to the multi-active multi-parasitic hybrid array operating in a multi-path channel.
	  Our proposed hybrid  architecture with parasitic elements outperforms conventional  architectures in terms of energy efficiency.
	\end{abstract}
	
	
	\section{Introduction}\label{sec: Introduction }

\subsection{Motivation and challenges}

Reconfigurable parasitic antennas have the potential to enable low-cost, energy-efficient beamforming for larger arrays~\cite{1141852,858918, kalis2014parasitic,10077503,6474484,7086418}. Unlike conventional arrays which adjust the excitation signal of each element, parasitic arrays use a combination of active and passive components that reduces the overall power consumption~\cite{papageorgiou2018efficient}. Parasitic elements are tuned using reconfigurable components such as varactor or PIN diodes. These elements interact with active currents through mutual coupling.
 The current on a parasitic element depends on the active antenna current via the reconfigurable load, complicating the beamforming gain optimization by creating non-linearity in the objective function.
We focus on developing a low-complexity approach for beamforming optimization, analyze the spectral efficiency and energy efficiency of a parasitic array system, and compare with conventional  architectures.

\subsection{Prior work}

 There are several types of reconfigurable antennas,  categorized by what property of the antenna is configured~\cite{7086418}. In this paper, we focus on configuring the radiation pattern through parasitic elements.  It is worth noting that there are other reconfigurable antenna types like dynamic metasurface antennas based on resonant frequency configuration~\cite{smith_analysis_2017}, movable antennas based on dynamically changing element spacing~\cite{10286328}, and polarization reconfigurable antennas~\cite{9723156} for different wireless applications. 
 Parasitic arrays can be used in different ways.  In \cite{6623074,9993759,6060873, 7432147}, parasitic elements are used  for direct symbol modulation and spatial multiplexing. In  \cite{6060884,sun2017mutual}, they are used for reducing mutual coupling in MIMO arrays.
In our paper,  we focus on using parasitics to improve the  beamforming gain.
  We envision use of parasitic elements to realize large arrays with reduced power consumption and hardware cost compared to conventional architectures.

Parasitic antennas have been widely studied by the antenna community.
While extensive research  has focused on designing and prototyping reconfigurable parasitic antennas~\cite{ 9886764, 10071988,nikkhah2013compact, 4463909, 8410614 }, the configuration methods predominantly rely on look-up tables that map 
target beamforming angles to 
PIN diode states.  These prototype results are limited to a specific design with fixed number of parasitic elements. 
The optimization metrics in \cite{9886764, 10071988, nikkhah2013compact} are limited to radiation pattern and realized gain valid only for a line-of-sight (LOS) channel.
Although \cite{4463909} and \cite{8410614}  provided system level throughput simulations for realistic channel models, the methods for configuring parasitic elements  are not scalable to large arrays because they require a brute-force search across all configurations.

Prior work  on configuring parasitic antennas in the wireless literature relies on  ideal abstractions of the radiation patterns without describing the underlying circuit configurations\cite{8424550, 9429908, 10404876}.
The physical realizability aspect is critical because the parasitic element configuration directly impacts the beamforming weights of the active antennas.
This is due to mutual coupling that causes the input impedance of the active antenna to change for each parasitic configuration. The total radiated power also changes as a function of the parasitic configuration. The active antenna current needs to be adjusted dynamically  to satisfy the power constraint.
Another limitation of prior work is that there are no closed-form expressions for the optimal weights.
Most prior work on configuring parasitic antennas is limited to iterative  methods such as
steepest gradient descent in \cite{966856}, stochastic  optimization in~\cite{papageorgiou2018efficient,4604722}, genetic algorithm in \cite{6060890}, quasi-Newton method in \cite{dardari2024dynamic},  greedy algorithm in \cite{8186257}, iterative mode selection in~\cite{8424550},  and multi-armed bandit theory-based radiation mode selection in~\cite{9429908, 10404876}. 
To summarize, \cite{papageorgiou2018efficient, 966856, 4604722,6060890,dardari2024dynamic,8186257,8424550,9429908,10404876}  fail to provide  insights into the optimal solutions for parasitic configurations under physically consistent constraints.

Multi-port circuit theory is a powerful tool for developing tractable and physically consistent models that account for  mutual coupling.
In prior work, this tool has been used for analyzing dense MIMO arrays\cite{wallace_mutual_2004, ivrlac_toward_2010}.
 The mutual coupling is modeled  in a mathematically tractable way with  impedance matrices~\cite{10373407}. There is limited prior work, however, on using this circuit-theory approach for analyzing parasitic arrays.
 In \cite{6877828}, a signal model was proposed for a parasitic array system with a single RF chain. This signal model was used to optimize a multi-user system in \cite{7037416}. Their approach focused on minimizing the normalized mean squared error between the desired precoder and the actual parasitic element precoder. In \cite{9286866,9969659},  the parasitic reactances were optimized using the correlation coefficient metric that aligns the actual radiation pattern to the desired one from maximum ratio transmission.
 These approaches in \cite{9286866,9969659} do not guarantee optimality from a spectral efficiency perspective because the radiated power constraint is not explicitly included. In \cite{6742715}, closed-form design guidelines were proposed for single-fed parasitic arrays to support arbitrary precoding. These guidelines were generalized to multi-active and multi-parasitic arrays in \cite{tariq2020design}. Although \cite{ 6742715, tariq2020design} provide closed-form expressions for the parasitic impedance, they are limited to design inequalities that the parasitic array should satisfy for the array to radiate. The circuit model has not been used to derive  closed-form solutions that optimize a communication-theoretic metric under realistic power constraints. We focus on deriving a closed-form solution not only because it has low-complexity but also to extract design insights from the solution and analyze scaling trends.

 \subsection{Contributions}

  We use the  multi-port circuit theory approach from \cite{ivrlac_toward_2010} to formulate the multiple-input-single-output (MISO) model for a reconfigurable parasitic array with multiple active and parasitic elements operating in a multi-path wireless channel.  We validate the derived theoretical beamforming gain expression through electromagnetic simulations which enables its use in formulating an optimization problem.
   We  solve the LOS channel beamforming optimization problem for the single active and multiple parasitic antenna case. Since the beamforming optimization problem is mathematically intractable in its original form, we propose a relaxed objective function that leads to a tractable closed-form expression for the reconfigurable parasitic reactance. We validate the accuracy of this solution by comparing it with the numerical approach from the MATLAB optimization toolbox.
  For the  beamforming optimization problem with a  hybrid multi-active multi-parasitic array operating in a multi-path channel, we propose a  near-optimal closed-form solution for the active antenna current  and the parasitic reactances for beamforming.
 We compare the spectral and energy efficiency of the hybrid parasitic array  beamforming with that of the fully digital array without any parasitic elements. We show that  parasitic elements improve the spectral efficiency  without any additional power overhead. This demonstrates that  beamforming with hybrid parasitic arrays is an energy-efficient approach.

			We highlight the novelty of our paper by comparing it to prior studies. The communication-theoretic approaches in \cite{8424550, 9429908, 10404876} employ realistic  system models, such as multi-path channels, but rely on generic abstractions of reconfigurable antenna radiation patterns and numerical optimization methods. The circuit-theoretic approaches in \cite{ 6742715, tariq2020design} utilize multi-port models for parasitic arrays, but are limited to providing closed-form design guidelines for parasitic impedances in the form of inequalities. They do not derive optimal solutions and are also restricted to LOS channels.
						In this paper, we formulate a beamforming optimization problem for a parasitic reconfigurable hybrid array operating in a multi-path propagation environment. We derive closed-form solutions for both the parasitic reactances and the active antenna currents, providing a  comprehensive and tractable framework for system analysis.

				\textit{Organization}:
				In Section~\ref{sec: Parasitic antenna preliminaries}, we use the circuit-theory approach to  formulate the system model for a parasitic array with single active antenna and multiple parasitic elements. We define the SNR and beamforming gain for a LOS channel. In Section~\ref{sec: Line-of-sight beamforming optimization  for a parasitic array with single active element}, we formulate the beamforming  optimization problem for the LOS channel and propose a closed-form expression for the parasitic reactance using an approximate  objective function.  The results from Section~\ref{sec: Line-of-sight beamforming optimization  for a parasitic array with single active element} are generalized to the case of multi-active, multi-parasitic hybrid array in a multi-path channel in Section~\ref{subsec: Multi-active and multi-parasitic beamforming optimization}. In Section~\ref{subsec: Numerical results }, we present numerical results on spectral efficiency and energy efficiency for the parasitic hybrid array, the fully active array, and hybrid phase-shifter sub-connected array  using the dipole antenna planar array implemented in  Feko software. In Section~\ref{sec: Conclusion and future work }, we summarize the key takeaways and discuss directions for future research. The MATLAB simulation code to generate the results is made publicly available to facilitate reproducibility~\cite{nitishcode}.
				
					\textit{Notation}:  A bold lowercase letter $\bm{\sfz}$ denotes a column vector, 	a bold uppercase letter $\bm{\sfZ}$ denotes a matrix, $|\cdot|$  indicates absolute value, $\angle(z)$ denotes argument of a complex number $z$, $(\cdot)^\rmT$ denotes transpose,  $(\cdot)^{\ast}$ denotes conjugate  transpose,  $(\cdot)^{\rmc}$ denotes conjugate, $\| .\|$ denotes the L2 norm, $\|.\|_\rmF$ denotes the Frobenius norm,	$\cR\{z\}$ denotes real part of a complex number $z$, $\cI\{z\}$ denotes  imaginary part of a complex number $z$, $ \bI_N$ represents the identity matrix of size $N$, $\floor*{x}$ is the floor function, and $\ceil*{x}$ is the ceil function, $\bm{\sfA} \kron \bm{\sfB}$ represents the Kronecker product of matrices $\bm{\sfA}$ and $\bm{\sfB}$, $\mathcal{N}_\bbC(0,1)$ represents the complex normal zero mean unit variance random variable.
				
				\section{System  model}\label{sec: Parasitic antenna preliminaries}
				
We use the multi-port circuit theory approach from \cite{wallace_mutual_2004,ivrlac_toward_2010, 10373407} to model the  reconfigurable parasitic array.
In Section~\ref{subsec: Multi-port circuit theory approach for antenna arrays}, we review the multi-port circuit modeling approach for an array with all active antennas similar to \cite{ivrlac_toward_2010}. In Section~\ref{subsec: Parasitic antenna array model }, we formulate the model for a parasitic array with only one active antenna. In Section~\ref{subsec: Parasitic array beamforming gain formulation for a line-of-sight channel}, we formulate the beamforming pattern in terms of the reconfigurable parasitic reactance.

				\subsection{Multi-port circuit theory model of a wireless communication system}\label{subsec: Multi-port circuit theory approach for antenna arrays}

								We first review the model for a conventional antenna array with all active antennas.	 We assume a narrowband MISO system with $N$ antennas at transmitter and single antenna at receiver as shown in Fig.~\ref{fig: system model and circuit model}(a). Each antenna at the transmitter is connected to an RF chain. All signals and multi-port parameters are a function of  frequency and evaluated at the center frequency $f_{\sfc}$ (or center wavelength $\lambda_{\sfc}$). 
					 The MISO system is modeled as an $N+1$ port network, as shown in Fig.~\ref{fig: circuit model}(a).
					In the circuit-theoretic formulation, 
					the EM interaction between different antennas is captured in the form of an impedance matrix. 
								The voltage at each antenna port is related to the currents of all antenna ports through the impedance matrix. 
				For a transmit array with $N$ elements, let the  impedance matrix of the transmit array be $\bm{\sfZ}_{\mathsf{TX}}\in \bbC^{N\times N}$.
				 Let the voltage vector evaluated at $f_{\sfc}$ be $\bm{\sfv}_{\mathsf{TX}}=[\mathsf{v}_0, \mathsf{v}_1, \dots, \mathsf{v}_{N-1} ]^\rmT\in \bbC^{N\times 1}$. Here, each voltage source is an abstraction of the signal at the output of the transmit RF chain.  Each RF chain is comprised of components like  a mixer, a local oscillator, filters, and amplifiers  that we do not show in Fig.~\ref{fig: system model and circuit model}(a) for brevity.
					Let the transmit information symbol in frequency domain be $s$ and we assume it has unit variance. The excitation signal at each antenna is  controlled using the digital beamformer.
				We use	the antenna current  vector to represent the  beamforming vector and denote it as $\bm{\sfi}_{\mathsf{TX}}=[\mathsf{i}_0, \mathsf{i}_1, \dots, \mathsf{i}_{N-1} ]^\rmT\in \bbC^{N\times 1}$. Hence, the transmit signal is $\bm{\sfi}_{\mathsf{TX}} s$.

					   \begin{figure}
					   	\centering
					   	\begin{subfigure}[t]{\linewidth}
					   		\centering
					   		\includegraphics[width=\linewidth]{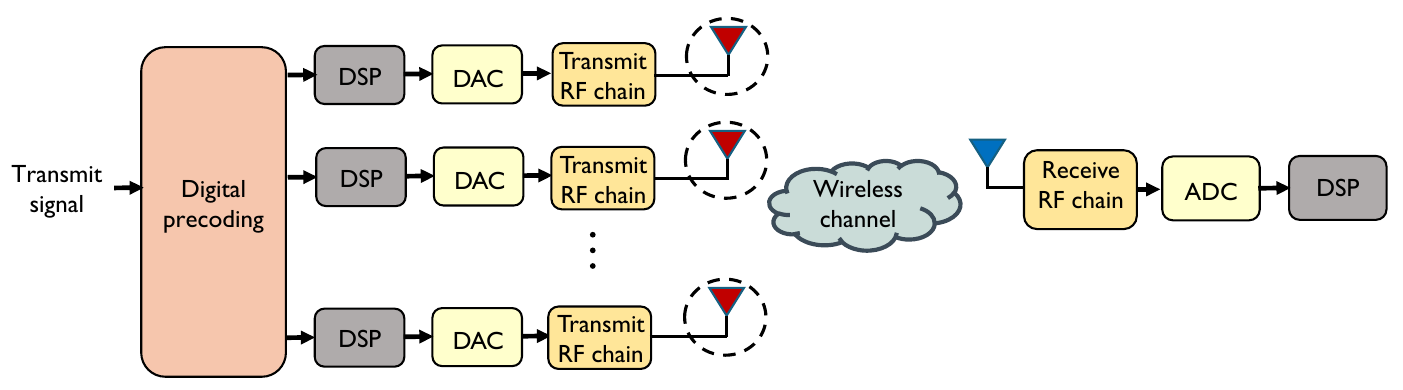}    
					   		\caption
					   		{
					   		}
					   		\label{fig: Conventional antenna array}
					   	\end{subfigure}
					   	\begin{subfigure}[t]{\linewidth}
					   		\centering
					   		\includegraphics[width=\linewidth]{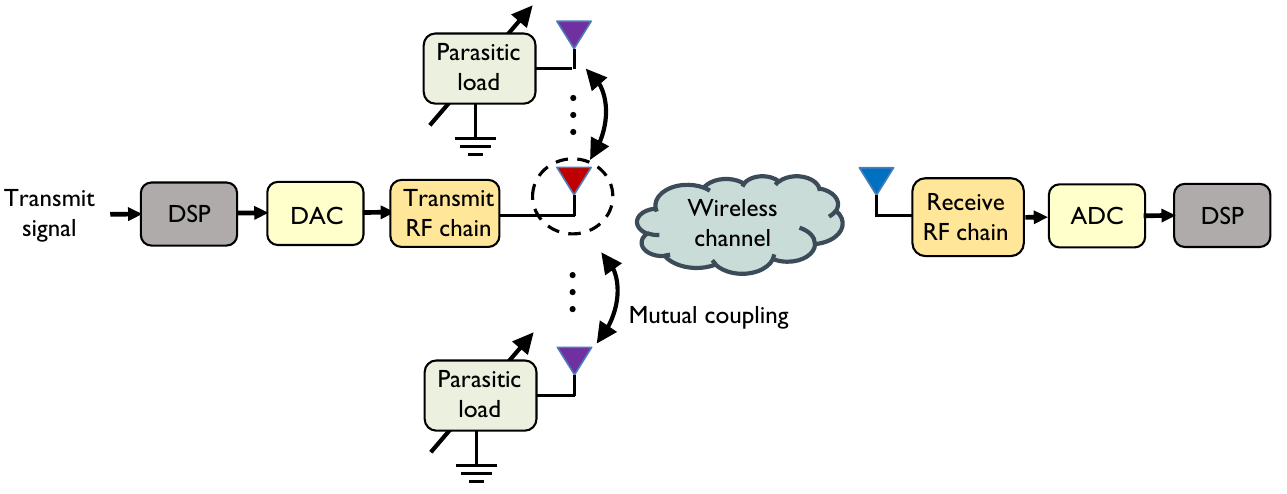}    
					   		\caption
					   		{
					   		}
					   		\label{fig: Parasitic antenna array}
					   	\end{subfigure}
					   	\caption{(a) In a conventional array, each antenna is connected to an RF chain (b) In a parasitic array, only one antenna is connected to an RF chain and the remaining  antennas are tuned by passive loads. 
					   	 }
					   	 \label{fig: system model and circuit model}
					   \end{figure}
					   
					   \begin{figure}
					   	\centering
					   	\begin{subfigure}[t]{0.49\linewidth}
					   		\centering
					   		\includegraphics[width=1\linewidth]{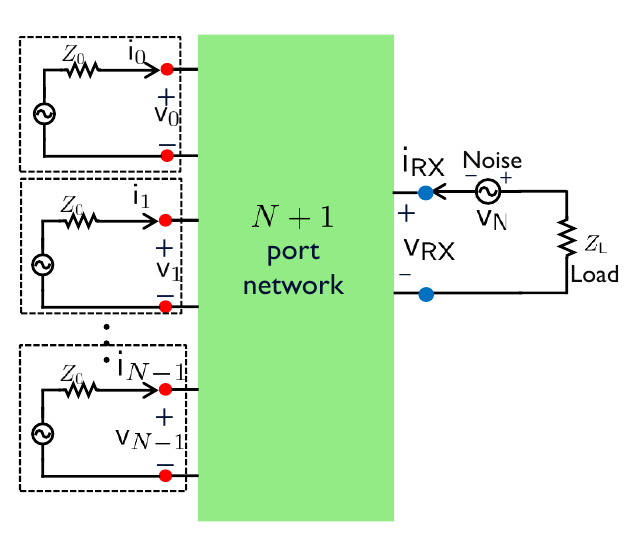}    
					   		\caption
					   		{
					   		}
					   		\label{fig: Multi-port circuit model for conventional array}
					   	\end{subfigure}
					   	\begin{subfigure}[t]{0.49\linewidth}
					   		\centering
					   		\includegraphics[width=1\linewidth]{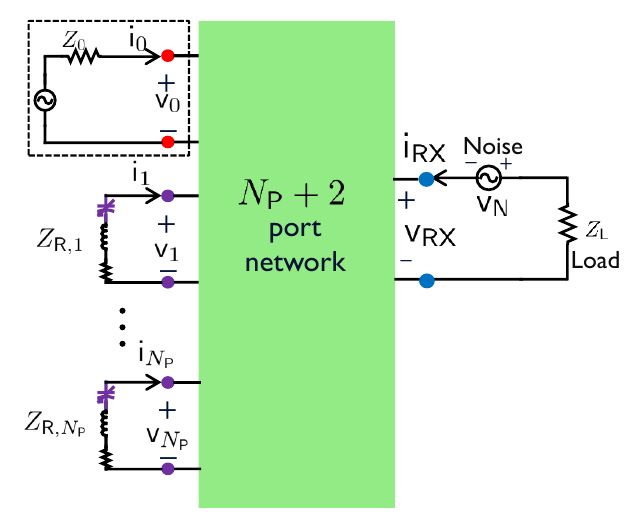}    
					   		\caption
					   		{
					   		}
					   		\label{fig: Multi-port circuit model for parasitic antenna array}
					   	\end{subfigure}
					   	\caption{ (a) Circuit model of a conventional array, where each antenna is modeled as a voltage source (b) Circuit model for a parasitic array, where the parasitic  elements are tuned using variable loads.
					   	}
					   	\label{fig: circuit model}
					   \end{figure}
				
				The circuit model relates the receive voltage to the transmit current.
			The impedance parameter of the single receive antenna is $\sfZ_{\mathsf{RX}}\in \bbC$. The voltage on the receiver antenna port is $\sfv_{\mathsf{RX}}\in \bbC$ and the current is $\sfi_{\mathsf{RX}}\in \bbC$.  The wireless propagation channel impedance vector $\bm{\sfz}_{\mathsf{RT}} = \bm{\sfz}_{\mathsf{TR}} \in \bbC^{N\times 1}$ accounts for the small-scale and large-scale fading in the wireless channel.
				From the multi-port model, we have~\cite{ivrlac_toward_2010}
\begin{subequations}
	\begin{equation}\label{eqn: vtx ztx itx x}
		\bm{\sfv}_{\mathsf{TX}} = 	\bm{\sfZ}_{\mathsf{TX}}  	\bm{\sfi}_{\mathsf{TX}} s + 	\bm{\sfz}_{\mathsf{TR}} 	\sfi_{\mathsf{RX}} ,
	\end{equation}
	\begin{equation}\label{eqn: v rx zrt itx x zrx irx}
		\sfv_{\mathsf{RX}}= \bm{\sfz}_{\mathsf{RT}}^\rmT \bm{\sfi}_{\mathsf{TX}} s  + \sfZ_{\mathsf{RX}} \sfi_{\mathsf{RX}}.
	\end{equation}
\end{subequations}
To simplify the  system model  by decoupling  $	\bm{\sfv}_{\mathsf{TX}}$ from $\sfi_{\mathsf{RX}}$, we
			 apply the unilateral approximation which results in 
			 $	\bm{\sfv}_{\mathsf{TX}}\approx \bm{\sfZ}_{\mathsf{TX}}  	\bm{\sfi}_{\mathsf{TX}} s$.
			 The unilateral approximation  is reasonable because 
			   the transmit array is
				in the far-field from the receive antenna such that the signal
				attenuation between them is large~\cite{ivrlac_toward_2010}.
				This means that the transmitter is unaffected by the
				electromagnetic fields at the receiver~\cite{ivrlac_toward_2010}.

				The circuit theory approach also  provides a physically consistent definition for the total transmit power radiated from the antenna array.  The total transmit power is~\cite{ivrlac_toward_2010}
				\begin{align}\label{eqn: trp}
					P_{\mathsf{TX}}= \mathbb{E}[ \cR  \{ (\bm{\sfi}_{\mathsf{TX}} s)^{\ast} \bm{\sfv}_{\mathsf{TX}}  \}],
				\end{align}
				where the expectation is over the random information symbol.
				Using the unilateral approximation and the assumption of the unit variance transmit information symbol, we obtain
				$P_{\mathsf{TX}}=\bm{\sfi}_{\mathsf{TX}} ^{\ast} \cR\{\bm{\sfZ}_{\mathsf{TX}}\}\bm{\sfi}_{\mathsf{TX}}$.
					The diagonal elements of the  $	\bm{\sfZ}_{\mathsf{TX}}$ matrix are called self-impedances while the non-diagonal elements are called mutual impedances.
 				Much prior work on MIMO assumes that $\bm{\sfZ}_{\mathsf{TX}}$ is a diagonal matrix, i.e.,   mutual impedance is negligible. This leads to the power being proportional to 
				  norm squared of the current vector. In this paper, we use the physically consistent definition of radiated power from \eqref{eqn: trp} that accounts for mutual coupling.

		We establish the input-output signal model by solving for the received voltage.  Let the load impedance on the receive antenna be $Z_{\sfL}$.  At the receiver, the voltage source $\sfv_{\sfN}$ models the noise due to background radiation. Applying Kirchhoff's voltage law at the receive port, we have $		\sfv_{\mathsf{RX}} = - \sfv_{\sfN}	-\sfi_{\mathsf{RX}} 	Z_{\sfL}.$
				 Substituting  $\sfv_{\mathsf{RX}}$ in \eqref{eqn: v rx zrt itx x zrx irx}, we have $	\sfi_{\mathsf{RX}} =  -\frac{\bm{\sfz}_{\mathsf{RT}}^\rmT \bm{\sfi}_{\mathsf{TX}} s  + 	\sfv_{\sfN} }{\sfZ_{\mathsf{RX}} + Z_{\sfL}}$.  Let the  voltage on the load impedance  $\sfv_{\sfL}$ be expressed in terms of the input  $\bm{\sfi}_{\mathsf{TX}} s$ as 
				  $ \sfv_{\sfL} = - \sfi_{\mathsf{RX}} Z_{\sfL} = \frac{Z_{\sfL} }{\sfZ_{\mathsf{RX}} + Z_{\sfL}}(\bm{\sfz}_{\mathsf{RT}}^\rmT \bm{\sfi}_{\mathsf{TX}} s + 	\sfv_{\sfN}) .$
			Assuming that the noise voltage source $\sfv_{\sfN} $	has a variance of $\sigma^2$ and $\mathbb{E}[|s|^2]=1$, the   $\mathsf{SNR}$ at the receiver is 
			\begin{align}\label{eqn: SNR def}
				\mathsf{SNR}&= \frac{|\bm{\sfi}_{\mathsf{TX}}^\rmT \bm{\sfz}_{\mathsf{RT}}|^2 }{\sigma^2}.
			\end{align}
			The $\mathsf{SNR}$ definition in \eqref{eqn: SNR def} is general and applies to both conventional arrays and parasitic arrays.

				\subsection{Parasitic antenna array model with one active element}\label{subsec: Parasitic antenna array model }

				 The array implementation dictates the constraints on the current vector. Beamforming in an antenna array is achieved by configuring the current vector $\bm{\sfi}_{\mathsf{TX}}$. In conventional active arrays, the magnitude and phase of each element in $\bm{\sfi}_{\mathsf{TX}}$ can be adjusted independently. As a result, the active antenna current vector is constrained only by maximum power. In contrast, parasitic arrays create beam  patterns using a combination of active phase control and mutual coupling with parasitic elements that have variable loads. These antennas, known as parasitic elements, are reconfigured by altering their variable loads. Changing the termination of these elements dynamically modifies their radiative properties, enabling beamforming but imposing additional constraints on $\bm{\sfi}_{\mathsf{TX}}$.

			In this section, 	we describe the model for a parasitic antenna array with one active antenna connected to an RF chain and $N_{\mathsf{P}}=N-1$ parasitic elements that are tuned using variable loads as shown in Fig.~\ref{fig: system model and circuit model}(b). The parasitic antenna array is characterized as an $N_{\mathsf{P}}+2$ port network as shown in Fig.~\ref{fig: circuit model}(b).  
			This representation is independent of the physical array geometry.
			The beam pattern, however, depends on the array structure described in Section~\ref{subsec: Parasitic array beamforming gain formulation for a line-of-sight channel}.
			Let the voltage at the port of the active antenna in the parasitic array be denoted as $	\mathsf{v}_0 $ and the current as $\mathsf{i}_0 $. The vector of voltages at the parasitic element ports is $	\bm{\sfv}_{\sfP}=[\mathsf{v}_1, \dots, \mathsf{v}_{N_{\mathsf{P}}} ]^\rmT$. We define $	\bm{\sfv}_{\mathsf{TX}} =[ 	\mathsf{v}_0,	\bm{\sfv}_{\sfP}^\rmT]^\rmT$.			
			 Similarly, the current  vector is  $	\bm{\sfi}_{\sfP}=[\mathsf{i}_1, \dots, \mathsf{i}_{N_{\mathsf{P}}} ]^\rmT$. Thus, $ \bm{\sfi}_{\mathsf{TX}} =[ 	\mathsf{i}_0,	\bm{\sfi}_{\sfP}^\rmT]^\rmT $.
The first element of these vectors $\bm{\sfi}_{\mathsf{TX}}$ and $\bm{\sfv}_{\mathsf{TX}}$ corresponds to the active antenna.

We relate the currents and voltages using impedance matrices.
			 Let the self-impedance of the active antenna be $	\sfz_{00}$. The vector of mutual impedances between the  active antenna and $N_{\mathsf{P}}$ parasitic elements is $\bm{\sfz}_{\sfm}$. The mutual impedance matrix of the $N_{\mathsf{P}}$ parasitic elements is $\bm{\sfZ}_{\mathsf{P}}$ and is fixed based on the array design. We use the block matrix notation to express $   \bm{\sfZ}_{\mathsf{TX}}=
			 \begin{bmatrix}
			 	\sfz_{00} &  \bm{\sfz}_{\sfm}^\rmT \\
			 	\bm{\sfz}_{\sfm}  &  \bm{\sfZ}_{\mathsf{P}}
			 \end{bmatrix} $.
			Using 	the unilateral approximation on  \eqref{eqn: vtx ztx itx x}, we have 		
			\begin{align}\label{eqn:  v Zi for single active multi parasitic}
				\begin{bmatrix}
					\mathsf{v}_0 \\
					\bm{\sfv}_{\sfP}
				\end{bmatrix}= \bm{\sfZ}_{\mathsf{TX}}	\begin{bmatrix}
					\mathsf{i}_0  s\\
					\bm{\sfi}_{\sfP} s
				\end{bmatrix}.
			\end{align}
			The magnitude and phase of $\mathsf{i}_0 $ can be controlled independently through the RF chain. 
			Of note, however, $\bm{\sfi}_{\sfP}$ cannot be tuned independently as the current on each parasitic element is induced from $\mathsf{i}_0 $ through coupling. 

			We now establish the dependence of $\bm{\sfi}_{\sfP}$  on $\mathsf{i}_0 $.
			Let the variable load at the port of the $\ell$th parasitic element be $Z_{\mathsf{R},\ell}$. We define the parasitic reconfigurable load matrix as  $\bm{\sfZ}_{\mathsf{R}}=\mathsf{diag}\{[Z_{\mathsf{R},1}, \dots, Z_{\mathsf{R},N_{\mathsf{P}}}]^\rmT\} $.  The voltage vector at the parasitic ports is  $	\bm{\sfv}_{\sfP} = -  \bm{\sfZ}_{\mathsf{R}}\bm{\sfi}_{\sfP} s .$
			From \eqref{eqn:  v Zi for single active multi parasitic} and expression of $	\bm{\sfv}_{\sfP}$,  the current vector on the parasitic ports in terms of the active antenna current, parasitic element mutual impedance matrix, and the reconfigurable load matrix is
			\begin{align}\label{ip in terms of Zp Zr and io}
				\bm{\sfi}_{\sfP} = -( \bm{\sfZ}_{\mathsf{P}} +\bm{\sfZ}_{\mathsf{R}} )^{-1}\bm{\sfz}_{\sfm} \mathsf{i}_0.
			\end{align}
			For a parasitic reconfigurable array, only current $\mathsf{i}_0$ is an independent variable.  All elements in the vector $\bm{\sfi}_{\sfP}$ are dependent on the reconfigurable load matrix $\bm{\sfZ}_{\mathsf{R}}$ which can be tuned independently, but constrains the set of possible values of $\bm{\sfi}_{\sfP}$.

					\subsection{Parasitic array beamforming gain formulation for a line-of-sight channel}
					\label{subsec: Parasitic array beamforming gain formulation for a line-of-sight channel}
					
					The mathematical formulation in Section~\ref{subsec: Multi-port circuit theory approach for antenna arrays} and Section~\ref{subsec: Parasitic antenna array model } is applicable for any general array with arbitrary  array geometry.
					In this section, we make some specific assumptions on the array geometry similar to the prior work for ease of comparison.
				We assume a linear antenna array oriented along the $x$ axis with inter-element spacing of $d$ as shown in Fig.~\ref{fig: linear array with parasitic elements}.
				The active antenna element is placed at the origin and the parasitic elements are arranged on both sides of the active element. The receive antenna is in the $xy$ plane at an angle $\theta$ from the $y$ axis. We assume a LOS channel between the transmit array and receive antenna. Let $\gamma$ be the coefficient that captures large-scale fading and
				 $		\ba_{\sfP}(\theta)$ be the parasitic antenna array steering vector				 
				 \begin{align}
				 	\ba_{\sfP}(\theta)&=\bigg[e^{-\sfj 2\pi \floor*{\frac{N_\sfP}{2}} \frac{ d\sin(\theta)}{\lambda_{\sfc}}}, \dots,  e^{-\sfj 2\pi \frac{d\sin(\theta)}{\lambda_{\sfc}}}, e^{\sfj 2\pi \frac{d\sin(\theta)}{\lambda_{\sfc}}} ,   \nonumber \\&\dots, e^{\sfj 2\pi \ceil*{\frac{N_\sfP}{2}} \frac{ d\sin(\theta)}{\lambda_{\sfc}}}  \bigg]^\rmT.
				 \end{align}
				 The phase offsets in $\ba_{\sfP}(\theta)$ are calculated relative to the active antenna which is located at the center of the linear array as shown in Fig.~\ref{fig: linear array with parasitic elements}.
			The wireless propagation channel is $	\bm{\sfz}_{\mathsf{RT}}= {\gamma}[1, \ba^\rmT_{\sfP}(\theta)]^\rmT$~\cite{ivrlac_toward_2010}.
				We express the $\mathsf{SNR}$ of a parasitic array for a LOS channel by using  \eqref{eqn: SNR def}  and \eqref{ip in terms of Zp Zr and io} to obtain 
					\begin{align}
						\mathsf{SNR}=  \frac{\gamma^2}{\sigma^2}   |\mathsf{i}_0|^2  |1  - \ba^\rmT_{\sfP}(\theta)( \bm{\sfZ}_{\mathsf{P}} +\bm{\sfZ}_{\mathsf{R}} )^{-1}\bm{\sfz}_{\sfm}|^2 .
					\end{align}
			For a parasitic array, we define the far-field beam pattern as 
			\begin{align}\label{eqn:  G theta Zr}
					\cG(\theta, \bm{\sfZ}_{\mathsf{R}})= |1  - \ba^\rmT_{\sfP}(\theta)( \bm{\sfZ}_{\mathsf{P}} +\bm{\sfZ}_{\mathsf{R}} )^{-1}\bm{\sfz}_{\sfm}|^2 .
			\end{align}
			The beam pattern is a complicated non-linear function of the reconfigurable load matrix
			$\bm{\sfZ}_{\mathsf{R}}$, unlike a conventional array where the beam pattern $|\bm{\sfi}_{\mathsf{TX}}^\rmT \bm{\sfz}_{\mathsf{RT}}|^2$ is based on a quadratic expression of the active antenna current vector.

					\begin{figure}
					\centering
					\includegraphics[width=0.25\textwidth]{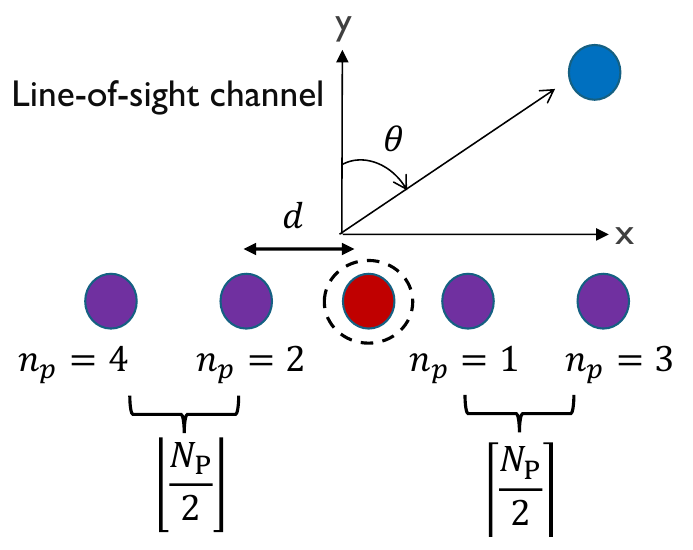}    
					\caption
					{Linear array consisting of one active antenna  and $N_\sfP$ parasitic elements  with a LOS channel to the receive antenna.
					}
					\label{fig: linear array with parasitic elements}
				\end{figure}

				\section{Line-of-sight beamforming   for a parasitic array with single active element}
				\label{sec: Line-of-sight beamforming optimization  for a parasitic array with single active element}
				

We  formulate the parasitic array beamforming problem and propose a simplification to the objective for ease of optimization in Section~\ref{subsec: Parasitic array beamforming problem formulation}. Using the proposed  simplification, we analyze the  parasitic beamforming weights and highlight the optimization challenge due to a coupled magnitude and phase relation in Section~\ref{subsec: Characterizing the parasitic beamforming weight matrix}. We propose a closed-form solution for the parasitic reactance in Section~\ref{subsec: Closed-form solution for the reconfigurable parasitic load reactance } and validate the analytical solution by comparing it with MATLAB based numerical approach in Section~\ref{subsec: Numerical results for LOS beamforming gain optimization}.
				
			\subsection{Parasitic array beamforming problem formulation}\label{subsec: Parasitic array beamforming problem formulation}

We develop an algorithm to adjust the weights to maximize the beamforming gain in the direction of $\theta_1$. This is optimal for a LOS channel.
	Reconfiguring  with varactor or PIN diodes is generally abstracted as terminating the parasitic element with reconfigurable reactive components\cite{nikkhah2013compact}.   
	We assume that the real part of each parasitic load is fixed and identical, i.e., $\cR\{\bm{\sfZ}_{\mathsf{R}}\}=\cR\{Z_{\sfR}\}\bI_{N_{\sfP}}$,  and that the imaginary part $\cI\{\bm{\sfZ}_{\mathsf{R}}\}$
			is reconfigurable.  
			These reactive components can be either inductors or capacitors\cite{papageorgiou2018efficient}.
			Hence, the elements of $\cI\{\bm{\sfZ}_{\mathsf{R}}\}$ can be positive or negative. 
			In practice, the elements of $\cI\{\bm{\sfZ}_{\mathsf{R}}\}$ are confined to a discrete set due to hardware limitations. Specifically, for varactor diodes, the reactive load values depend on the DC biasing voltage.  These constraints result in a finite and discrete set of achievable reactance values. For theoretical analysis, we simplify this by assuming $\cI\{\bm{\sfZ}_{\mathsf{R}}\}$ can take continuous values. This approximation closely mimics what is possible with high-resolution tuning hardware while offering analytical tractability.
	The transmitter  maximizes  \eqref{eqn:  G theta Zr} at    $\theta_1$ by configuring the parasitic reactances as		
				\begin{alignat}{3}
				\textbf{P1: }	& \cI\{ \bm{\sfZ}_{\mathsf{R}}^\star\}  =\underset{\cI\{ \bm{\sfZ}_{\mathsf{R}}\} \in \bbR}{\mbox{ argmax }} \cG(\theta_1, \bm{\sfZ}_{\mathsf{R}}).
			\end{alignat}
			Prior approaches to solve optimizations similar to	\textbf{P1} involve using  iterative approaches~\cite{papageorgiou2018efficient,tariq2020design}. We propose a closed-form solution that yields design insights by analyzing the dependence of the parasitic reactance on parameters like the mutual impedance vector, array steering vector, and self-impedance of the parasitic element.
				
				
				The beam pattern has a complicated dependence on the reconfigurable load $\bm{\sfZ}_{\mathsf{R}}$.	
				The matrix $\bm{\sfZ}_{\mathsf{R}}$ is inside an inverse term in \eqref{eqn:  G theta Zr} which makes it difficult to solve \textbf{P1} in its current form.
			In a conventional array beamsteering problem, the pattern depends on a linear combination of beamforming weights.
			To obtain a linear relation in \eqref{eqn:  G theta Zr}, we  simplify the matrix inverse via   a diagonal approximation.
			From antenna array literature, we know that the self-impedance  is generally higher in magnitude compared to the mutual impedance~\cite{balanis2015antenna}. With this observation, we separate all self-impedances of the parasitic array by defining $\bm{\sfD}_\sfP=\mathsf{diag}\{ \bm{\sfZ}_{\mathsf{P}} \}$ and the hollow matrix $\bm{\sfE}_\sfP=\bm{\sfZ}_{\mathsf{P}}- \bm{\sfD}_\sfP$ which includes all mutual impedances.
			The matrix $\bm{\sfZ}_{\mathsf{R}}$ is diagonal based on the parasitic array design defined in Section~\ref{subsec: Parasitic antenna array model }.
			We combine the two diagonal matrices $\bm{\sfZ}_{\mathsf{R}}$ and $\bm{\sfD}_\sfP$ by defining the parasitic beamforming weight matrix $	\bW$ which is also diagonal as
			\begin{equation}\label{eqn:  W inv}
				\bW^{-1}=\bm{\sfD}_\sfP+\bm{\sfZ}_{\mathsf{R}}.
			\end{equation} 
			With these definitions, we write the inverse term in \eqref{eqn:  G theta Zr} as $(\bm{\sfZ}_{\mathsf{P}} +\bm{\sfZ}_{\mathsf{R}} )^{-1} = (\bm{\sfE}_\sfP + \bm{\sfD}_\sfP +\bm{\sfZ}_{\mathsf{R}})^{-1} = (\bm{\sfE}_\sfP + 	\bW^{-1} )^{-1}$.
			For small $N_\sfP$,	as $\|  \bm{\sfE}_\sfP\|_\rmF$ is small relative to $\|\bW^{-1}  \|_\rmF $, we further expand  the term	$(\bm{\sfE}_\sfP + \bW^{-1}  )^{-1}$  using perturbation theory as
			\begin{equation}\label{eqn: perturb theory approx}
				(\bm{\sfE}_\sfP + \bW^{-1}   )^{-1} =  \bW - \bW \bm{\sfE}_\sfP \bW + \cO(\| \bm{\sfE}_\sfP \|^2_\rmF).
			\end{equation}
			Using the first order term from \eqref{eqn: perturb theory approx}, we approximate the objective $\cG(\theta_1, \bm{\sfZ}_{\mathsf{R}})$ by defining $\widehat{\cG}(\theta_1, \bW) $ in terms of $\bW$  as
			\begin{equation}\label{G hat}
				\widehat{\cG}(\theta_1, \bW)=|1  - \ba^\rmT_{\sfP}(\theta_1)\bW\bm{\sfz}_{\sfm}|^2.
			\end{equation}
				The approximate objective 	$\widehat{\cG}(\theta_1, \bW)$ is easier to optimize because  $\bW$ is a diagonal matrix. 
			This motivates us to use $\widehat{\cG}(\theta_1, \bW)$ as the new objective function instead of $\cG(\theta_1, \bm{\sfZ}_{\mathsf{R}})$.  In Appendix~\ref{app: proof of approximation}, we show that $\cG(\theta_1, \bm{\sfZ}_{\mathsf{R}}) \rightarrow \widehat{\cG}(\theta_1, \bW) $ for small $\|  \bm{\sfE}_\sfP\|_\rmF$.

	\subsection{Coupled magnitude and phase constraint}\label{subsec: Characterizing the parasitic beamforming weight matrix}
	
	In a conventional active array beamforming problem, the magnitude and phase of the beamforming weight can be tuned independently. 
	For parasitic arrays, we analyze the  beamforming weight matrix $\bW$ and show that the phase and magnitude are coupled that makes the optimization problem challenging.
The $i$th   diagonal element of the matrix $\bW$ is $	[\bW]_{ii}=\frac{1}{[\bm{\sfD}_\sfP]_{ii}+[\bm{\sfZ}_{\mathsf{R}}]_{ii}}.$
		The phase  of $	[\bW]_{ii}$ is 
		\begin{align}\label{eqn: varphi i def}
		\angle([\bW]_{ii})&= \mathsf{atan2}(\cI\{[\bW]_{ii}\}, \cR\{[\bW]_{ii}\}),\\ \nonumber&\stackrel{(a)}{=}\mathsf{tan}^{-1}\left(\frac{-(\cI\{[\bm{\sfD}_\sfP]_{ii}\}+\cI\{[\bm{\sfZ}_{\mathsf{R}}]_{ii}\}  )}{(\cR\{[\bm{\sfD}_\sfP]_{ii}\}+\cR\{[\bm{\sfZ}_{\mathsf{R}}]_{ii}\} )}\right),
		\end{align}
		where $(a)$ follows from the fact that the real part is always positive as it is a resistance term.
	We define $\zeta_i=\frac{1}{(\cR\{[\bm{\sfD}_\sfP]_{ii}\}+\cR\{[\bm{\sfZ}_{\mathsf{R}}]_{ii}\}  )} $ and $ \varphi_i = 	\angle([\bW]_{ii})$. The beamforming weight $	[\bW]_{ii}$  is expressed in terms of $\varphi_i$ and $\zeta_i$ as 
	\begin{equation}\label{eqn:  Wii in terms of phase}
		[\bW]_{ii}=\zeta_i\cos(\varphi_i)e^{\sfj \varphi_i}.
	\end{equation}
%
	From \eqref{eqn:  Wii in terms of phase},  the magnitude of 		$[\bW]_{ii}$  varies as a sinusoidal function of the phase angle,  which complicates the beamforming problem.

We visualize the beamforming weight 	$[\bW]_{ii}$ in the complex domain by plotting its locus.
We  assume that the parasitic elements are identical. Hence, 
	the self-impedances of all parasitic elements are same, i.e.,  $[\bm{\sfZ}_{\mathsf{P}}]_{ii}=[\bm{\sfD}_{\mathsf{P}}]_{ii}=Z_\sfP$.	
	We define $\zeta= (  \cR\{Z_\sfP\}  + \cR\{Z_{\sfR}\} )^{-1}$ and set $\zeta_i=\zeta~\forall~i$.
		We plot the locus of the normalized beamforming weight ${[\bW]_{ii} }/{\zeta}$ in complex domain  as a function of the phase angle $\varphi_i$ in Fig.~\ref{fig: bf gain complex response}. The locus of ${[\bW]_{ii}}/{\zeta}$ is a circle with center at $\left(0.5,0\right)$
and a radius of $0.5$.  The phase angle $\varphi_i$ has a limited phase range $[-\pi/2, \pi/2]$.
This constraint  is known as the Lorentzian constraint. 
 It also appears in the  literature on  dynamic metasurface antennas (DMA)~\cite{bowen2022optimizing,10584442,deshpande_qif1}.

We reformulate the optimization problem based on the Lorentzian constraint observation to leverage the solution approach developed in \cite{bowen2022optimizing,10584442,deshpande_qif1}.
From \eqref{eqn: varphi i def} and the definition of $\zeta$, we can express the  optimization variable $\cI\{[\bm{\sfZ}_{\mathsf{R}}]_{ii}\}$ in terms of the phase angle $	\angle([\bW]_{ii})$ as
\begin{align}\label{eqn:  imag Zr ii}
\cI\{[\bm{\sfZ}_{\mathsf{R}}]_{ii}\}=- \cI\{Z_\sfP\}-\frac{\tan(	\angle([\bW]_{ii}))}{\zeta}.
\end{align}
In \eqref{eqn:  imag Zr ii}, we see that there is a one-to-one mapping between $\cI\{[\bm{\sfZ}_{\mathsf{R}}]_{ii}\}$ and the approximate beamforming weight phase angle.
	We formulate the beamforming problem using the approximate objective $\widehat{\cG}(\theta, \bW)$  as 
	\begin{subequations}
	\begin{alignat}{3}
		\textbf{P2: }	& \underset{\bW}{\mbox{ max }} \widehat{\cG}(\theta_1, \bW),\\
			&\text{ s.t. } {[\bW]_{ii}}/{\zeta}=\cos(\varphi_i)e^{\sfj \varphi_i},\label{eqn: cos varphi i}\\
			& \varphi_i \in \left[-{\pi}/{2}, {\pi}/{2}\right].
	\end{alignat}
\end{subequations}
The constraints of problem 	\textbf{P2} are similar to the beamforming constraints  of a DMA ~\cite{bowen2022optimizing,10584442,deshpande_qif1}. This inspires us to  use an optimization solution approach similar to~\cite{deshpande_qif1}.

In \cite{deshpande_qif1},  the DMA beamforming weight is transformed from a unit-modulus beamforming weight to simplify the optimization problem. We express the normalized parasitic beamforming weight in terms of a unit-modulus beamforming weight $e^{\sfj{\phi}_i}$ by applying a shift-of-origin transformation as
\begin{align}\label{eqn:  0.5 1+e j phi}
\frac{[\bW]_{ii}}{\zeta}=\frac{	1}{2}(1+e^{\sfj{\phi}_i}).
\end{align}
Equating the magnitude and phase of the right hand side of \eqref{eqn: cos varphi i} and \eqref{eqn:  0.5 1+e j phi}, we obtain the relation between ${\phi}_i$ and $\varphi_i$ as
\begin{align}\label{eqn: phi i 2 varphi}
	\phi_i= 2 \varphi_i = 2 	\angle([\bW]_{ii}).
\end{align}
The geometric interpretation of \eqref{eqn:  0.5 1+e j phi} and  \eqref{eqn: phi i 2 varphi} is shown in Fig.~\ref{fig: bf gain complex response}.

As $\phi_i$ has range $[-\pi, \pi]$, it is easier to solve problem 
	\textbf{P2} in terms of $\phi_i$ instead of $\varphi_i$.  Let $\boldsymbol{\Phi}=\mathsf{diag}\{[e^{\sfj\phi_1}, \dots, e^{\sfj\phi_{N_\sfP}}]^\rmT\}$.  We express  $\bW$ in terms of $\boldsymbol{\Phi}$ as 
	\begin{align}\label{eqn: bW in terms of Phi}
		\bW= \frac{\zeta}{2}(\bI_{N_\sfP}+ \boldsymbol{\Phi})		.
	\end{align}
	By substituting $\bW$ from \eqref{eqn: bW in terms of Phi} in \eqref{G hat},  we obtain
	\begin{align}\label{eqn:  Ghat in terms of Phi}
	\widehat{\cG}(\theta_1, \boldsymbol{\Phi})=\left|1  - \frac{\zeta}{2}\ba^\rmT_{\sfP}(\theta_1)(\bI_{N_\sfP}+ \boldsymbol{\Phi})\bm{\sfz}_{\sfm} \right|^2.
	\end{align}
The  problem 
\textbf{P2} is then equivalent to

	\begin{subequations}
	\begin{alignat}{3}
		\textbf{P3: }	& \underset{\boldsymbol{\Phi}}{\mbox{ max }} \widehat{\cG}(\theta_1, \boldsymbol{\Phi}),\\
		&\text{ s.t. } \angle([\boldsymbol{\Phi}]_{ii})\in [-\pi, \pi]\label{eqn: phi range}.
	\end{alignat}
\end{subequations}
We solve 	\textbf{P3} in terms of $\boldsymbol{\Phi}$ and then obtain the closed-form  expression of parasitic reactance using \eqref{eqn:  imag Zr ii} and  \eqref{eqn: phi i 2 varphi}.

\begin{figure}
	\centering
	\includegraphics[width=0.5\textwidth]{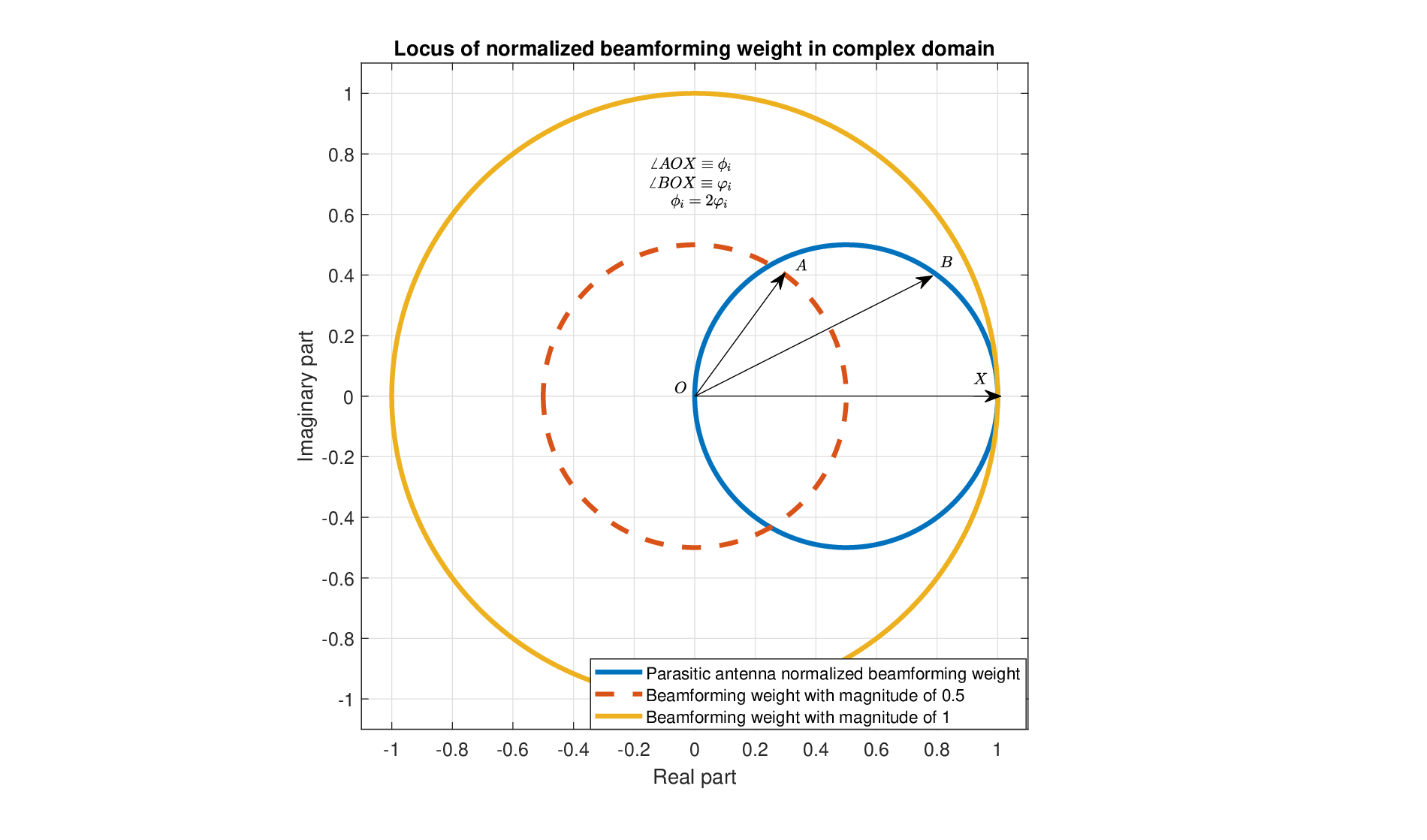}    
	\caption
	{The locus of the parasitic beamforming weight in the complex plane is a circle centered at $(0.5, 0)$ with radius $0.5$. We use a shift-of-origin transformation that reformulates the optimization in terms of the unit magnitude beamforming weight to avoid dealing with the coupled magnitude and phase constraint.
	}
	\label{fig: bf gain complex response}
\end{figure}

\subsection{Closed-form solution for the reconfigurable parasitic load reactance}\label{subsec: Closed-form solution for the reconfigurable parasitic load reactance }

We follow a mathematical approach similar to \cite{deshpande_qif1} to derive the closed-form solution for $\boldsymbol{\Phi}$. In the following theorem, we provide the closed-form solution to \textbf{P3}.
\begin{theorem}\label{thm: closed-form solution for phase}
	The closed-form solution to problem \textbf{P3} is 
		\begin{align}\label{eqn:  Phi i optimal}
		\angle([\boldsymbol{\Phi}]_{ii})\!=\!\angle\bigg(\!\! 1\!-\!\frac{\zeta\ba^\rmT_{\sfP}(\theta_1)\bm{\sfz}_{\sfm}}{2}\!\!\!\bigg)\!-\!  \angle\left([\ba_{\sfP}(\theta_1)]_i  [\bm{\sfz}_{\sfm}]_i\right)\!+\!(2m\!+\!1)\pi\!,
	\end{align}
	where $m$ is any integer.
\end{theorem}
\begin{proof}
	See Appendix~\ref{proof: thm: closed-form solution for phase} for the proof.
\end{proof}

   From the closed-form expression of $\boldsymbol{\Phi}$ in \eqref{eqn: Phi i optimal}, the term $ \angle\left( 1-{\zeta}\ba^\rmT_{\sfP}(\theta_1)\bm{\sfz}_{\sfm}/{2}\right) $ is common across all parasitic elements. In standard beamforming problems, the optimal solution is unaffected by common phase shifts~\cite{balanis2015antenna}. For beamforming architectures like parasitic antennas or DMAs, which have coupled magnitude and phase constraints as in \eqref{eqn: cos varphi i}, the reformulation introduces an offset term in the objective function that is independent of the beamforming weights, as shown in \eqref{eqn: Ghat in terms of Phi}. This parasitic constraint alters the problem, causing the optimal phase angle solution to depend on common phase rotations.

We  derive a closed-form expression that explicitly shows the dependence of the $i$th element's   load reactance on the  array steering vector, mutual impedance vector, self-impedance of the parasitic element, and real part of the parasitic load.
\begin{theorem}\label{thm: Imag reconfig reactance}
The closed-form solution for the optimal reactance obtained from the solution to \textbf{P3} is
\begin{align}\label{eqn: T zr closed form}
	&\cI\{[\bm{\sfZ}_{\mathsf{R}}]_{ii}\}=- \cI\{Z_\sfP\}-\frac{\cot\left(\frac{\angle([\ba_{\sfP}(\theta_1)]_i  [\bm{\sfz}_{\sfm}]_i)-\angle(   1-\frac{\zeta}{2}\ba^\rmT_{\sfP}(\theta_1)\bm{\sfz}_{\sfm})}{2}\right)  }{\zeta}.
\end{align}
\end{theorem}
\begin{proof}
	Using \eqref{eqn:  imag Zr ii}, \eqref{eqn: phi i 2 varphi}, and  \eqref{eqn:  Phi i optimal}, and   $\tan\left(  \frac{2m+1}{2}\pi-\theta\right)=\cot(\theta)~\forall~m \in \bbZ$, we obtain the desired result.
	\end{proof}

From \eqref{eqn: T zr closed form}, we see that the first term in the expression of $\cI\{[\bm{\sfZ}_{\mathsf{R}}]_{ii}\}$ is negative of the self reactance of the parasitic element. 
We observe that this first  term is the same for all parasitic elements.
This means  the proposed solution can be implemented with a simple hardware design for the reconfigurable load that has a common fixed reactance to cancel out the parasitic antenna self reactance connected in series with a tunable reactance  given by the second term in \eqref{eqn: T zr closed form} which depends on the steering vector.

	The closed-form parasitic reactance expression is optimal for $N_{\sfP}=1$ because the approximation error $\widetilde{\cG}(\theta, Z_{\mathsf{R}})=0$. 
For $N_{\sfP}>1$, the closed-form expression in \eqref{eqn: T zr closed form} is sub-optimal for problem  \textbf{P1} because  the optimization is based on the approximate objective $\widehat{\cG}(\theta, \bW)$.
In the numerical results section, we show that the proposed closed-form solution is nearly-optimal for $N_{\sfP}=2$ by comparing it with the solution obtained from the MATLAB Global search toolbox. As the value of $N_{\sfP}$ increases beyond 2, the approximation is worse. For practical designs, however, it suffices to use $N_{\sfP}=2$ as shown in Section~\ref{subsec: Numerical results }.

\subsection{Numerical results for LOS beamforming  optimization}\label{subsec: Numerical results for LOS beamforming gain optimization}

In this section, we evaluate the performance of the proposed parasitic beamforming approach for a dipole array for different values of $N_{\sfP}$.  The simulation parameters  are as follows: The center frequency $f_\sfc=7$ GHz, the radius of the dipole antenna is $\frac{\lambda_\sfc}{500}$, and length is $\frac{\lambda_\sfc}{2}$. The inter-element spacing is $0.4 \lambda_\sfc$. The dipole array impedance matrix is obtained similar to our prior work~\cite{deshpande2023analysis} which used expressions of mutual and self impedance from \cite{balanis2015antenna}.

In Fig.~\ref{fig: Np4 Np2 Np1 d 0.4}, we plot the maximum beamforming pattern ${\mbox{ max }}_{\cI\{ \bm{\sfZ}_{\mathsf{R}}\} } \cG(\theta_1, \bm{\sfZ}_{\mathsf{R}})$ as a function of the  angle $\theta_1$ for $N_\sfP \in \{1, 2, 4\}$. We compare the proposed closed-form approach in \eqref{eqn: T zr closed form} to that obtained by numerically solving problem 	\textbf{P1} using MATLAB Global Search function. For $N_{\sfP}=1$, both solutions coincide as shown in Fig.~\ref{fig: Np4 Np2 Np1 d 0.4}.
For $N_\sfP=2$ and $N_\sfP=4$, the performance of the proposed approach closely follows the numerical approach.  The main benefit of the proposed approach is that it is easier to compute compared to the solution obtained from the  iterative global search.
An interesting feature of the parasitic array architecture is that an increase in the beamforming pattern maximum is  obtained by only increasing $N_\sfP$ and not adding any extra active antenna. This gain also translates to spectral efficiency gain which we demonstrate in Section~\ref{subsec: Numerical results }.

\begin{figure}
	\centering
	\includegraphics[width=0.5\textwidth]{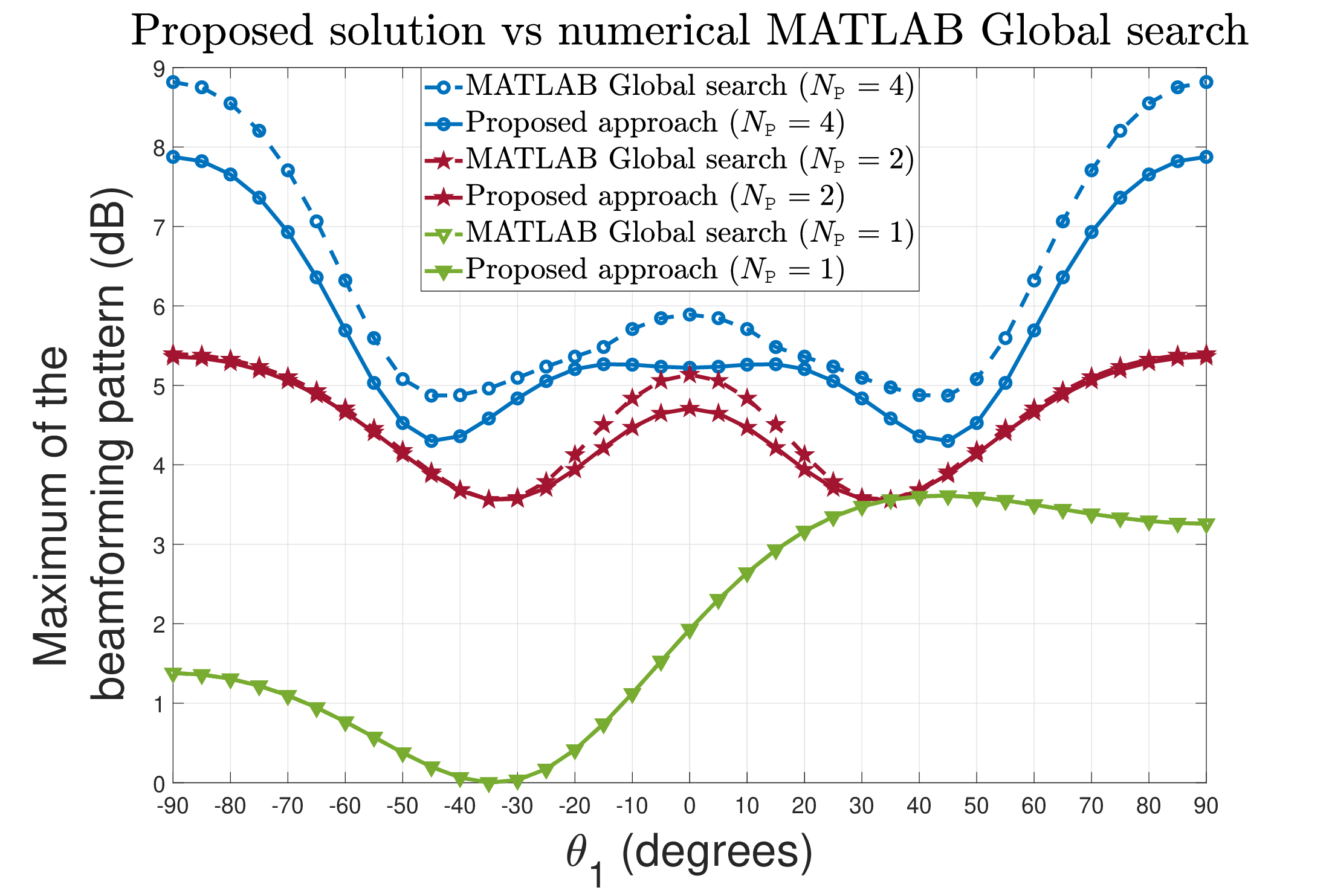}    
	\caption
	{The maximum value of the  beamforming pattern  at each  angle is shown for different values of $N_\sfP$. The proposed closed-form approach is exactly same as the MATLAB global search for $N_\sfP=1$ and approximately similar for $N_\sfP=\{2,4\}$.
	}
	\label{fig: Np4 Np2 Np1 d 0.4}
\end{figure}


The parasitic array exhibits some unique characteristics because of its unconventional beamforming approach and array configuration.
The array geometry, i.e., location of the parasitic element with respect to the active antenna and the number of parasitic elements impacts the beamforming pattern characteristics.
For example, the symmetricity of the maximum beamforming pattern plot depends on whether $N_{\sfP}$ is even or odd.
From Fig.~\ref{fig: Np4 Np2 Np1 d 0.4}, we observe that for $N_\sfP=4$, the  beamforming pattern is higher for angles close to $90^{\circ}$  compared to  angles close to $0^{\circ}$. This is because of the superdirectivity effect observed at endfire angle for  a dipole array~\cite{deshpande2023analysis}. We leverage this observation in  the design of a planar array with multiple active and parasitic elements in Section~\ref{subsec: Multi-active and multi-parasitic beamforming optimization}.

				\section{Beamforming with hybrid multi-active and multi-parasitic array in a  multi-path channel}\label{subsec: Multi-active and multi-parasitic beamforming optimization}
				
				A hybrid architecture combining multiple active antennas with parasitic elements offers a  practical solution to adding parasitics.
			It overcomes the problem that 
				 as inter-element spacing grows, the mutual coupling effect diminishes, leading to reduced performance gains. This architecture leverages the benefits of active antennas for improved control and parasitic elements for cost efficiency, making them well-suited for large-scale deployments.
								In this section, we generalize the circuit model derived previously to the case of  multi-active and multi-parasitic array. Then, we formulate the SNR maximization problem under radiated power constraint. Finally, we derive a low-complexity closed-form solution for  jointly optimizing active antenna currents and parasitic reactances in Section~\ref{sec: low complexity closed form soln}.

				\subsection{Multi-port circuit model and channel model generalization}\label{sec: multiport ckt model chan mod gen}
				

We extend the multi-port circuit model from Fig.~\ref{fig: circuit model},  to now accommodate multiple active antennas. Let the array consist of $N_\sfA$ active antennas, each paired with $N_\sfP$ parasitic elements, resulting in a total of $N_\sfA N_\sfP$ parasitic elements. 	Let the $N_\sfA\times~1$ active antenna current vector be $	\bm{\sfi}_\sfA=[\mathsf{i}_{\sfA,1}, \dots, \mathsf{i}_{\sfA,N_\sfA} ]^\rmT$ and $N_\sfA N_\sfP\times 1$ parasitic element current vector be 
$	\bm{\sfi}_{\sfP}=[\mathsf{i}_{1,1}, \dots, \mathsf{i}_{N_{\mathsf{P}}, 1}, \dots,  \mathsf{i}_{1,N_{\sfA}}, \dots, \mathsf{i}_{N_\sfP,N_{\sfA}}  ]^\rmT$. The current beamforming vector across all active and parasitic elements is denoted as $\bm{\sfi}_{\mathsf{TX}}=[\bm{\sfi}_\sfA^\rmT, \bm{\sfi}_{\sfP}^\rmT]^\rmT $.

  The mutual impedance matrix of dimension $ N_\sfP N_{\sfA} \times N_{\sfA} $ between the parasitic elements and active antennas is $\bm{\sfZ}_{\sfm}$. 
We denote the mutual coupling vector between the $j$th active antenna and the $N_\sfP$ parasitic elements belonging to the row of the $i$th active antenna as $\bm{\sfz}_{\sfm, i, j} \in \bbC^{N_\sfP \times 1}$.
We express $\bm{\sfZ}_{\sfm}$ in terms of the block-matrix notation as 
\begin{align}
	\bm{\sfZ}_{\sfm} = \begin{bmatrix}
		\bm{\sfz}_{\sfm, 1, 1} &\dots &	\bm{\sfz}_{\sfm, 1, N_\sfA}\\
		\vdots & \ddots & \vdots \\
		\bm{\sfz}_{\sfm, N_\sfA, 1} &  \dots  &\bm{\sfz}_{\sfm, N_\sfA, N_\sfA}
	\end{bmatrix}.
\end{align}
The matrix $\bm{\sfZ}_{\sfA}$ of dimension $  N_{\sfA} \times N_{\sfA} $ is the mutual impedance between all active antennas and  matrix $\bm{\sfZ}_{\sfP}$ of dimension $ N_\sfP  N_{\sfA} \times  N_\sfP N_{\sfA} $  is the mutual impedance between all parasitic antennas. 
We combine these matrices into a single mutual impedance matrix for the hybrid array and denote as $ \bm{\sfZ}_{\mathsf{TX}}$. Using block matrix notation, we have $	\bm{\sfZ}_{\mathsf{TX}}=
\begin{bmatrix}
	\bm{\sfZ}_{\sfA} &  \bm{\sfZ}_{\sfm}^\rmT \\
	\bm{\sfZ}_{\sfm}  &  \bm{\sfZ}_{\mathsf{P}}
\end{bmatrix} .$
Let the variable load at the port of the
$i$th parasitic element corresponding to the $j$th active antenna be 
${Z_{\mathsf{R},i, j}}$ and let $\bm{\sfZ}_{\mathsf{R}, j,j}=\mathsf{diag}\{[Z_{\mathsf{R},1, j}, \dots, Z_{\mathsf{R},N_{\mathsf{P}}, j}]\}$.
The parasitic reconfigurable load matrix is generalized as 
\begin{align}\label{eqn: ZR def}
	\bm{\sfZ}_{\mathsf{R}}\!=\!\mathsf{diag}\{[\underbrace{Z_{\mathsf{R},1, 1}, \dots, Z_{\mathsf{R},N_{\mathsf{P}}, 1}}_{\mathsf{diag}\{\bm{\sfZ}_{\mathsf{R}, 1,1}\}}, \dots, \underbrace{Z_{\mathsf{R},1, N_\sfA}, \dots, Z_{\mathsf{R},N_{\mathsf{P}}, N_\sfA}}_{\mathsf{diag}\{ \bm{\sfZ}_{\mathsf{R}, N_\sfA,N_\sfA}  \}}]^\rmT\!\}\!.
\end{align}
Following the steps similar to \eqref{eqn:  v Zi for single active multi parasitic}, we express the parasitic beamforming current vector as
\begin{align}\label{eqn: iP for multi active multi passive}
	\bm{\sfi}_{\sfP} = -( \bm{\sfZ}_{\mathsf{P}} +\bm{\sfZ}_{\mathsf{R}} )^{-1}\bm{\sfZ}_{\sfm}  	\bm{\sfi}_\sfA.
\end{align}
From \eqref{eqn: iP for multi active multi passive}, we see that the parasitic array beamforming vector depends on the reconfigurable load matrix $\bm{\sfZ}_{\mathsf{R}}$ and the active antenna current beamforming vector $\bm{\sfi}_\sfA$, both of which can be independently configured.

					\begin{figure}
						\centering
						\includegraphics[width=0.5\textwidth]{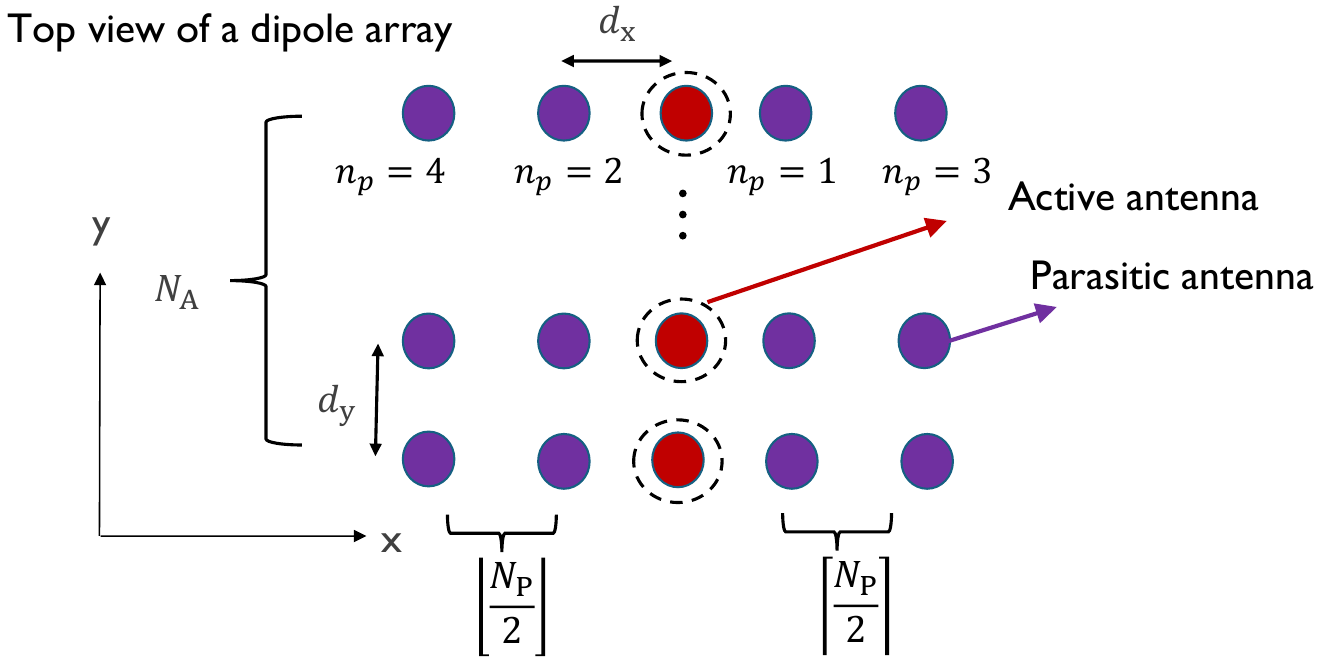}    
						\caption
						{Planar array with $N_\sfA$ active antennas (red) and $N_\sfP N_\sfA$ parasitic elements (purple).
						}
						\label{fig: planar array}
					\end{figure}

					 We now generalize to the geometric multi-path channel with a planar array from	Fig.~\ref{fig: planar array}.
					  In this work, we only consider steering in the $xy$ plane.
					 To define the channel model with planar array, we use the following definition of array steering vectors along $x$ and $y$ directions.
					 \begin{subequations}
					 	\begin{align}
					 			\ba_{\sfx}(\theta)&=\bigg[e^{-\sfj 2\pi \floor*{\frac{N_\sfP}{2}} \frac{ d_\sfx\sin(\theta)}{\lambda_{\sfc}}}, \dots,  e^{-\sfj 2\pi \frac{d_\sfx\sin(\theta)}{\lambda_{\sfc}}}, e^{\sfj 2\pi \frac{d_\sfx\sin(\theta)}{\lambda_{\sfc}}} ,   \nonumber \\&\dots, e^{\sfj 2\pi \ceil*{\frac{N_\sfP}{2}} \frac{ d_\sfx\sin(\theta)}{\lambda_{\sfc}}}  \bigg]^\rmT,
					 	\end{align}
					 		\begin{align}
					 				\ba_{\sfy}(\theta)= \bigg[1, e^{-\sfj 2\pi \frac{d_\sfy\cos(\theta)}{\lambda_{\sfc}}},   \dots, e^{-\sfj 2\pi \frac{(N_\sfA-1)d_\sfy\cos(\theta)}{\lambda_{\sfc}}} \bigg]^\rmT.
					 		 	\end{align}	
					 \end{subequations}
					 The geometric narrowband multi-path channel has $L$ paths where the complex gain of the $\ell$th path is $\alpha_\ell$ and the angle of incidence is $\theta_\ell$.
					 Let $\bh_\sfA=[h_{\sfA,1}, \dots, h_{\sfA,N_\sfA}]^\rmT$ be the channel vector corresponding to the active antennas
					which  are along $y$ axis. We define $\bh_\sfA=\frac{1}{\sqrt{L}}\sum_{\ell=1}^L \alpha_{\ell}\ba_{\sfy}(\theta_\ell)$.
					Similarly, let $\bh_\sfP=[\bh_{\sfP,1}^\rmT,\dots, \bh_{\sfP,N_\sfA}^\rmT]^\rmT$ be the channel vector corresponding to the parasitic elements. We define $\bh_\sfP = \frac{1}{\sqrt{L}}\sum_{\ell=1}^L \alpha_{\ell}(\ba_{\sfy}(\theta_\ell)\kron 	~\ba_{\sfx}(\theta_\ell) )$.
We define the channel vector for the hybrid parasitic array as  $\bh = [\bh_\sfA^\rmT, \bh_\sfP^\rmT]^\rmT$.
Including large-scale fading, we have $\bm{\sfz}_{\mathsf{RT}}= {\gamma}\bh$.
					 For $L=1$ and $N_\sfA=1$, we can verify that  $\bh$ simplifies to the  case of LOS channel with single active and multiple parasitic elements.

\subsection{SNR maximization under radiated power constraints}\label{sec: SNR max problem}

	In this section,	we  express the $\mathsf{SNR}$ defined in \eqref{eqn: SNR def} for the case of multi-active and multi-parasitic array with a multi-path channel  
		in terms of the active antenna current vector $\bm{\sfi}_\sfA$ and the reconfigurable load matrix $\bm{\sfZ}_{\mathsf{R}}$. 
		Using  $\bm{\sfi}_{\mathsf{TX}}$ and $ \bm{\sfz}_{\mathsf{RT}}$, we have 
		\begin{align}
			\mathsf{SNR}(\bm{\sfi}_\sfA, \bm{\sfZ}_{\mathsf{R}})
&= \frac{|\bm{\sfi}_{\mathsf{TX}}^\rmT \bm{\sfz}_{\mathsf{RT}}|^2 }{\sigma^2}
			, \\
			&\stackrel{(a)}{=} \frac{\gamma^2}{\sigma^2}|\bm{\sfi}_\sfA^\rmT \underbrace{(\bh_\sfA  - \bm{\sfZ}_{\sfm}^\rmT ( \bm{\sfZ}_{\mathsf{P}} +\bm{\sfZ}_{\mathsf{R}} )^{-1} \bh_\sfP  )}_{\bh_{\mathsf{eff}}}|^2,\label{eqn: heff def}
		\end{align}
				where equality $(a)$ follows from \eqref{eqn: iP for multi active multi passive}.
				The effective wireless channel $\bh_{\mathsf{eff}}=[h_{\mathsf{eff},1}, \dots, h_{\mathsf{eff},N_\sfA}]^\rmT$ experienced by the active antenna array depends on the  multi-path channel, mutual impedance between matrices, and reconfigurable parasitic load matrix.  This effective channel is dynamically tuned through the matrix $\bm{\sfZ}_{\mathsf{R}}$. 

			The transmit current vector is constrained to limit the total radiated power.
				 The total transmit power radiated from the antenna array is expressed using \eqref{eqn: trp} as 
				 \begin{equation}\label{eqn: P tx iA iP}
			P_{\mathsf{TX}}= \begin{bmatrix}
				\bm{\sfi}_\sfA^{\ast} &  \bm{\sfi}_{\sfP}^{\ast}
			\end{bmatrix}\begin{bmatrix}
				\cR\{\bm{\sfZ}_{\sfA}    \}  &  \cR\{\bm{\sfZ}_{\sfm}^\rmT  \}\\
				\cR\{\bm{\sfZ}_{\sfm}  \} &  \cR\{ \bm{\sfZ}_{\sfP} \}
			\end{bmatrix}\begin{bmatrix}
				\bm{\sfi}_\sfA\\
				\bm{\sfi}_{\sfP}
			\end{bmatrix}.
				 \end{equation}
				 To express the power in terms of $	\bm{\sfi}_\sfA$ only, we define an effective impedance matrix $\bm{\sfZ}_{\mathsf{eff}}$ as
				 	\begin{align}\label{eqn: Z eff def}
				 	&\bm{\sfZ}_{\mathsf{eff}}=\cR\{\bm{\sfZ}_{\sfA} \} + \bm{\sfZ}_{\sfm}^\ast ( \bm{\sfZ}_{\mathsf{P}} +\bm{\sfZ}_{\mathsf{R}} )^{-\ast}\cR\{ \bm{\sfZ}_{\sfP} \} (\bm{\sfZ}_{\mathsf{P}} +\bm{\sfZ}_{\mathsf{R}} )^{-1}\bm{\sfZ}_{\sfm}    \nonumber \\&-\cR\{ \bm{\sfZ}_{\sfm}^\rmT\}(\bm{\sfZ}_{\mathsf{P}} +\bm{\sfZ}_{\mathsf{R}} )^{-1}\bm{\sfZ}_{\sfm}-    \bm{\sfZ}_{\sfm}^\ast ( \bm{\sfZ}_{\mathsf{P}} +\bm{\sfZ}_{\mathsf{R}} )^{-\ast}  \cR\{\bm{\sfZ}_{\sfm}  \}.
				 \end{align}
				 	Using \eqref{eqn: iP for multi active multi passive}, \eqref{eqn: P tx iA iP}, and \eqref{eqn: Z eff def}, we express $	P_{\mathsf{TX}}$ in terms of $	\bm{\sfi}_\sfA$ and   $\bm{\sfZ}_{\mathsf{eff}}$ as $P_{\mathsf{TX}}=\bm{\sfi}_\sfA^{\ast}  \bm{\sfZ}_{\mathsf{eff}}\bm{\sfi}_\sfA$. 
				The effective impedance matrix $\bm{\sfZ}_{\mathsf{eff}}$ is a function of the reconfigurable load matrix and mutual impedance matrices. This effective matrix is dynamically tuned through $\bm{\sfZ}_{\mathsf{R}}$. For a conventional array with active elements only,  $\bm{\sfZ}_{\mathsf{eff}}= \cR\{\bm{\sfZ}_{\sfA} \}$.

				We now formulate the optimization problem with the  objective as the $\mathsf{SNR}$ and a constraint of $P_{\mathsf{max}}$ on the total radiated power from the transmit array.   The optimization variables include  the active antenna beamforming vector $	\bm{\sfi}_\sfA$ and the imaginary part of the parasitic reconfigurable load matrix $\cI\{ \bm{\sfZ}_{\mathsf{R}}\}$.  The  beamforming optimization problem for the hybrid architecture is mathematically expressed as 
					\begin{subequations}\label{eqn: prob SNR opt}
					\begin{alignat}{3}
						\textbf{P4: }\bm{\sfi}_\sfA^{\star},  \cI\{ \bm{\sfZ}_{\mathsf{R}}^\star\} &=	 \underset{\bm{\sfi}_\sfA \in \bbC,\cI\{ \bm{\sfZ}_{\mathsf{R}}\} \in \bbR}{\mbox{ argmax }} 	\mathsf{SNR}(\bm{\sfi}_\sfA, \bm{\sfZ}_{\mathsf{R}}),\\
						&\text{ s.t. } \bm{\sfi}_\sfA^{\ast}   \bm{\sfZ}_{\mathsf{eff}}\bm{\sfi}_\sfA \leq P_{\mathsf{max}}.\label{eqn: power constraint}
					\end{alignat}
				\end{subequations}
	The problem 		\textbf{P4} is challenging to solve directly because the optimization variables are tightly coupled in both objective and constraints.

				\subsection{Low-complexity closed-form solution for \textbf{P4}}\label{sec: low complexity closed form soln}
				
	In this section, we propose a low-complexity closed-form approach for configuring $\bm{\sfi}_\sfA $ and $\cI\{ \bm{\sfZ}_{\mathsf{R}}\} $  that uses an approximation  similar to Section~\ref{sec: Line-of-sight beamforming optimization  for a parasitic array with single active element}. We define an upper bound on the original objective as $	\mathsf{SNR}_{\mathsf{up}}(\bm{\sfi}_\sfA, \bm{\sfZ}_{\mathsf{R}})$  using the triangle inequality as 
	\begin{align}\label{eqn: SNR up}
\mathsf{SNR}_{\mathsf{up}}(\bm{\sfi}_\sfA, \bm{\sfZ}_{\mathsf{R}}) = \frac{\gamma^2}{\sigma^2} \left(\sum_{\ell=1}^{N_\sfA}  |\mathsf{i}_{\sfA,\ell}| |h_{\mathsf{eff},\ell}|\right)^2.
	\end{align}
The optimization in \textbf{P4} with respect to $\cI\{ \bm{\sfZ}_{\mathsf{R}}\}$ is approximately solved by maximizing the upper bound $\mathsf{SNR}_{\mathsf{up}}(\bm{\sfi}_\sfA, \bm{\sfZ}_{\mathsf{R}})$.
The term $h_{\mathsf{eff},\ell}$ depends on the parasitic load reactances of all elements as defined in \eqref{eqn: heff def}.  
We use a diagonalization approximation similar to Section~\ref{subsec: Parasitic array beamforming problem formulation}. By defining $\bm{\sfD}_\sfP=\mathsf{diag}\{ \bm{\sfZ}_{\mathsf{P}} \}$, we  compute an approximate expression for the $\ell$th element of the effective channel $\hat{h}_{\mathsf{eff},\ell}$ expressed as 
\begin{align}
	\hat{h}_{\mathsf{eff},\ell} = h_{\sfA,\ell}\left(1- \frac{	\bm{\sfz}_{\sfm, \ell, \ell}^\rmT(\bm{\sfD}_{\sfP, \ell , \ell}  + \bm{\sfZ}_{\mathsf{R}, \ell ,\ell } )^{-1} \bh_{\sfP,\ell} }{h_{\sfA,\ell}} \right).
\end{align}
By using the approximation $\hat{h}_{\mathsf{eff},\ell}$ in  \eqref{eqn: SNR up}, we decouple the parasitic reactance optimization problem into $N_\sfA$ sub-problems as maximizing the upper bound is obtained by maximizing each summand. The solution to the $\ell$th sub-problem  $\cI\{ \bm{\sfZ}^{\star}_{\mathsf{R}, \ell, \ell}\}$ is mathematically expressed as
	\begin{align}
		\textbf{P5: }	 \cI\{ \bm{\sfZ}^{\star}_{\mathsf{R}, \ell, \ell}\}= \underset{ \cI\{ \bm{\sfZ}_{\mathsf{R}, \ell, \ell}\}}{\mbox{ argmax }} | 	\hat{h}_{\mathsf{eff},\ell} |^2, \forall~\ell=\{1, \dots, N_\sfA\}.
	\end{align}
Each sub-problem in 	\textbf{P5} can be expressed in the form similar to  problem \textbf{P2} because the matrix $(\bm{\sfD}_{\sfP, \ell , \ell}  + \bm{\sfZ}_{\mathsf{R}, \ell ,\ell } )^{-1}$ is a diagonal matrix.  This enables us to  apply the closed-form solution proposed in Theorem~\ref{thm: Imag reconfig reactance} to this problem as well.

The mathematical form of each sub-problem in \textbf{P5} resembles that of problem \textbf{P1}. Hence, the closed-form solution to problem \textbf{P1} proposed in Theorem~\ref{thm: Imag reconfig reactance} is directly applied to each sub-problem in \textbf{P5}.  After substituting the closed-form solution of $\cI\{ Z_{\mathsf{R}, i,j}^{\star}  \}$ in the objective and constraints of problem \textbf{P4}, we optimally solve for the active antenna beamforming vector $	\bm{\sfi}_\sfA$.
	The  closed-form solution to problem 	\textbf{P4} based on the solution of the approximate sub-problems in 	\textbf{P5} is denoted as $\bm{\sfi}_\sfA^{\star}$  and $Z_{\mathsf{R}, i,j}^{\star}  $. The imaginary part of $Z_{\mathsf{R}, i,j}^{\star}  $ is 
\begin{align}\label{eqn: Izr star closed-form}
	&\cI\{ Z_{\mathsf{R}, i,j}^{\star}  \}= - \cI\{Z_\sfP\} \nonumber \\ &
	-\frac{\cot\left(\frac{{\angle\left(\frac{[\bh_{\sfP,j}]_i  [\bm{\sfz}_{\sfm, j,j}]_i}{h_{\sfA,j}}\right)-\angle\left(   1-\frac{\zeta}{2 h_{\sfA,j}}\bh^\rmT_{\sfP,j}\bm{\sfz}_{\sfm, j, j}\right)}}{2}	  \right)}{\zeta}.
\end{align}
By substituting $\cI\{ Z_{\mathsf{R}, i,j}^{\star}  \}$ in \eqref{eqn: ZR def}, \eqref{eqn: heff def}, and \eqref{eqn: Z eff def}, we denote $ 	\bm{\sfZ}_{\mathsf{R}}^{\star}$, $\bh_{\mathsf{eff}}^{\star}$, and $\bm{\sfZ}_{\mathsf{eff}}^{\star}$.  The closed-form expression for $\bm{\sfi}_\sfA^{\star}$ is 
\begin{align}\label{eqn: iA star}
	\bm{\sfi}_\sfA^{\star}=  \sqrt{P_{\mathsf{max}}} \frac{(\bm{\sfZ}_{\mathsf{eff}}^{\star})^{-1}{\bh_{\mathsf{eff}}^{\star}}^\rmc}{\|(\bm{\sfZ}_{\mathsf{eff}}^{\star})^{-1/2}{\bh_{\mathsf{eff}}^{\star}}^\rmc     \|}.
\end{align}

The solution for the active antenna beamforming current in \eqref{eqn: iA star} is optimal  for the problem \textbf{P4} assuming that 
	$	\bm{\sfZ}_{\mathsf{R}}$,  $\bm{\sfZ}_{\mathsf{eff}}$, and $\bh_{\mathsf{eff}}$ are fixed. While deriving a closed-form expression for the parasitic reactance in \eqref{eqn: Izr star closed-form}, several approximations are made, leading to a sub-optimal solution. The active antenna beamforming vector, however,  compensates for this sub-optimality by leveraging its optimal closed-form solution in \eqref{eqn: iA star}.  
	This solution strategy is consistent with prior work on hybrid beamforming using conventional antennas, as seen in \cite{ayach_spatially_2014}. In the next section, we compare the proposed beamforming approach with various benchmarks.

				\section{Spectral efficiency  and energy efficiency numerical results  }\label{subsec: Numerical results }
				
				\subsection{Radiation pattern validation of parasitic array with Feko}\label{subsec: Validation of the radiation pattern of parasitic array with Feko simulations}

				In this section, we validate the theoretical  expression of the parasitic array $\mathsf{SNR}$ from \eqref{eqn: heff def} for a LOS channel based on the circuit theory approach through computational electromagnetic simulations in the  software Feko.			
				  In our simulations, we use thin wire dipole antenna for both active and parasitic elements. The planar array configuration  in Fig.~\ref{fig: planar array} is implemented  with $N_\sfA=6$, $N_\sfP=2$, $d_\sfx=0.4 \lambda_{\sfc}$, and $d_\sfy=0.5 \lambda_{\sfc}$.  The simulation is performed at $f_\sfc=7$ GHz.
					In Feko, we use active voltage sources for the six antennas placed on the $y$ axis. Each active antenna is connected to a voltage source with unit magnitude and zero phase angle, i.e., $	\bm{\sfv}_{\sfA}=[1,1,1,1,1,1]^\rmT$. For the remaining 12 dipole antennas, we attach reconfigurable loads.  The resistance of each load is assumed to be a fixed  and set to a small value of  $0.05~\Omega$ to avoid significant losses.  The reactances are chosen arbitrarily from the set of real numbers. We implement two different configurations for the parasitic reactances. 
				For configuration 1, the reactances are set to $\cI\{ \mathsf{diag}\{\bm{\sfZ}_{\mathsf{R}}\}\}=[-10, 20, 40, 100, 300, 0, 5, 15, 45, 70, -60, -90]^\rmT \Omega$ and for configuration 2, they are set to $\cI\{ \mathsf{diag}\{\bm{\sfZ}_{\mathsf{R}}\}\}=[0, 200, -150, 40, -10, -50, 60, 120, -30, 40, 0, 10]^\rmT \Omega$.
				
				For the theoretical calculations, we extract the $18\times 18$ scattering matrix $\bS_{\mathsf{TX}}$ at 7 GHz and a reference impedance of $50~\Omega$ from Feko. We convert the scattering matrix to impedance matrix using the relation $	\bm{\sfZ}_{\mathsf{TX}}=50 (\bI_{18}-\bS_{\mathsf{TX}})^{-1}(\bI_{18}+\bS_{\mathsf{TX}})$.
				The matrices $\bm{\sfZ}_{\mathsf{A}}$, $\bm{\sfZ}_{\sfm}$ and $\bm{\sfZ}_{\mathsf{P}}$ are obtained through $	\bm{\sfZ}_{\mathsf{TX}}$.  For validating beam pattern in the $xy$ plane, we use $\bh$ with $L=1$ and compute $\bh_\sfP $ and $\bh_\sfA $ as described in Section~\ref{sec: multiport ckt model chan mod gen}.
				We compute $\bm{\sfi}_{\sfA}$ using $\bm{\sfv}_{\sfA}$, $\bm{\sfZ}_{\mathsf{A}}$, $\bm{\sfZ}_{\sfm}$, $\bm{\sfZ}_{\mathsf{P}}$, and $\bm{\sfZ}_{\mathsf{R}}$ as 
				$\bm{\sfi}_{\sfA}=(\bm{\sfZ}_{\mathsf{A}}  - \bm{\sfZ}_{\sfm}^\rmT(\bm{\sfZ}_{\mathsf{P}} +\bm{\sfZ}_{\mathsf{R}} )^{-1} \bm{\sfZ}_{\sfm} )^{-1}\bm{\sfv}_{\sfA}$. Substituting $\bm{\sfi}_{\sfA}$, $\bm{\sfZ}_{\sfm}$, $\bm{\sfZ}_{\mathsf{P}}$,  $\bm{\sfZ}_{\mathsf{R}}$, $\bh_\sfP $, and $\bh_\sfA $ in \eqref{eqn: heff def}, we obtain the beamforming  pattern.  For comparison purposes, we normalize this pattern such that the maximum value across the angular domain corresponds to 0 dB.
				In Fig.~\ref{fig:config feko}(a) and  Fig.~\ref{fig:config feko}(b), we show the comparison of the theoretical  pattern and the simulated  pattern from Feko for two different parasitic reactance configurations.
				For both configurations, we see that the simulated and theoretical beam patterns are well aligned. This validates the theoretical expression derived in \eqref{eqn: heff def}. A benefit of using the proposed model is that it does not require updated simulation results each time the parasitic reactances change. The matrices $\bm{\sfZ}_{\mathsf{A}}$, $\bm{\sfZ}_{\sfm}$ and $\bm{\sfZ}_{\mathsf{P}}$ are independent of the load configurations and depend only on the antenna and array geometry.
				For the further simulations, we will use the impedance matrices
				$\bm{\sfZ}_{\mathsf{A}}$, $\bm{\sfZ}_{\sfm}$ and $\bm{\sfZ}_{\mathsf{P}}$ obtained from Feko and use the  expression in  \eqref{eqn: heff def} to optimize $\bm{\sfi}_\sfA$ and $ \cI\{ \bm{\sfZ}_{\mathsf{R}}\}$.

				\begin{figure}
					\centering
					\begin{subfigure}[t]{0.5\linewidth}
						\centering
						\includegraphics[width=1\linewidth]{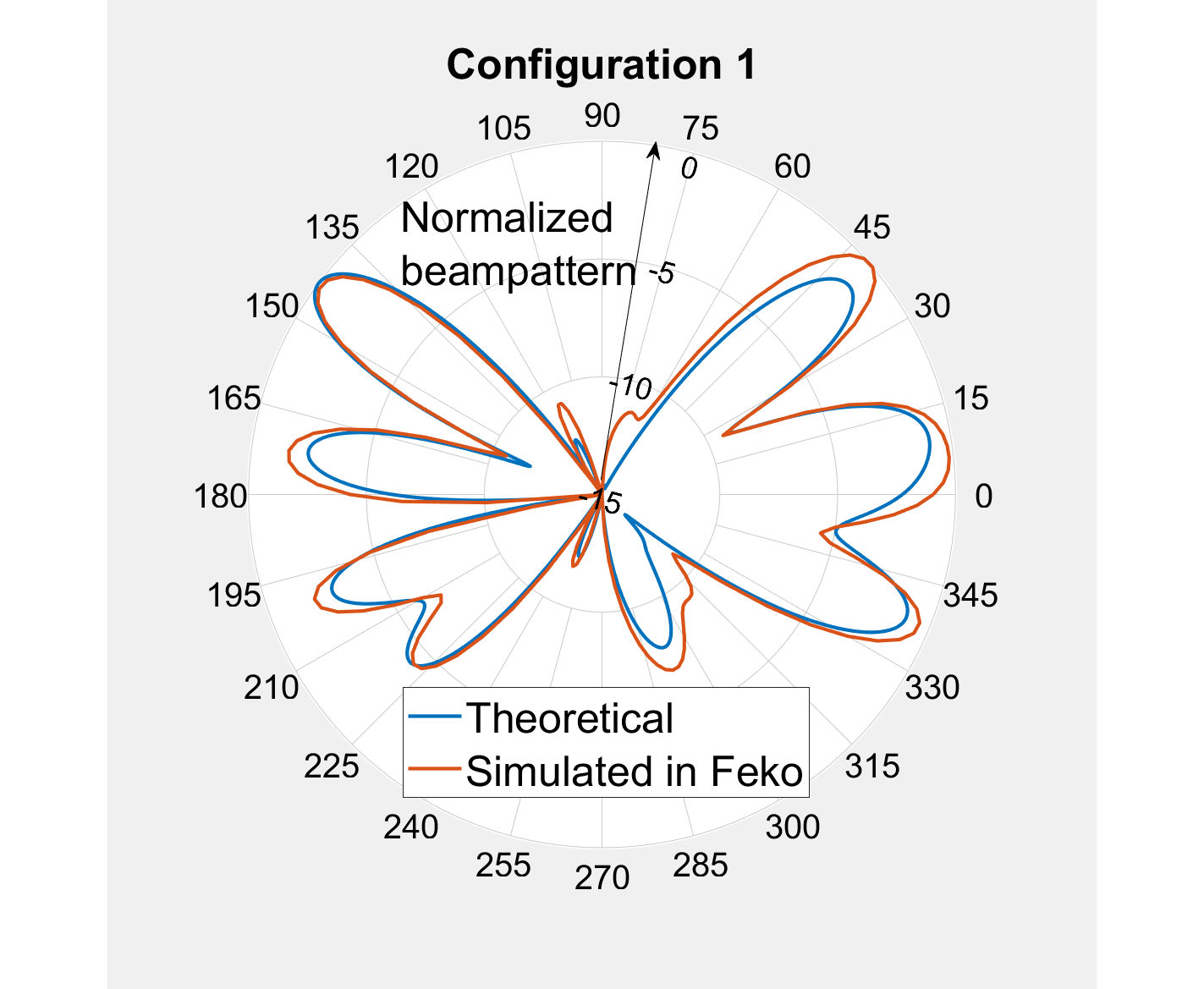}    
						\caption
						{
						}
						\label{fig: config 1}
					\end{subfigure}~\hfil
					\begin{subfigure}[t]{0.5\linewidth}
						\centering
						\includegraphics[width=1\linewidth]{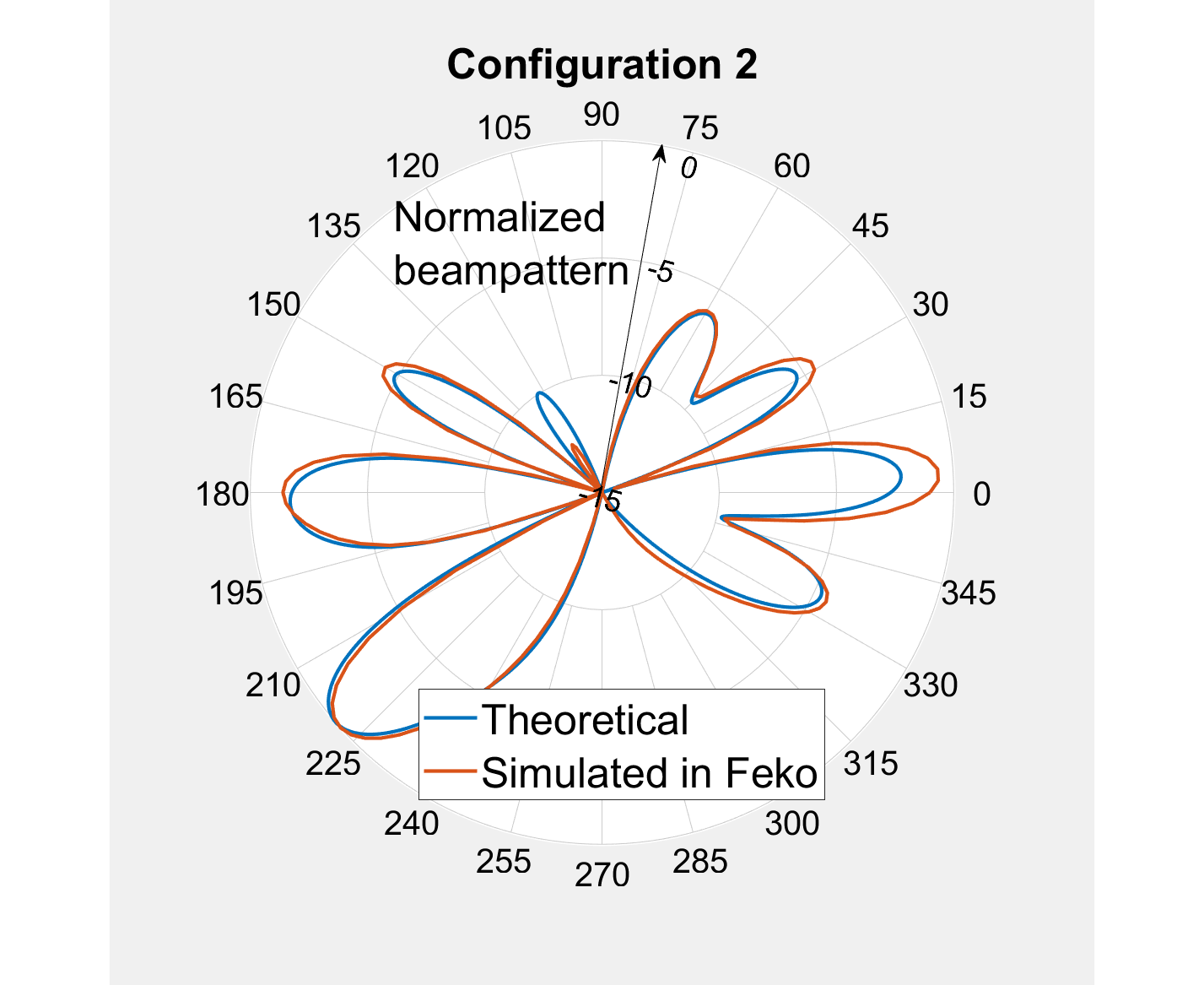}    
						\caption
						{
						}
						\label{fig:config 2}
					\end{subfigure}					
					\caption{The normalized beamforming pattern is shown as a function of the azimuth angle for two different parasitic configurations.
						The theoretical beamforming  pattern closely matches that of the  Feko simulations and thus demonstrates the validity of the circuit theory approach for modeling parasitic arrays.}\label{fig:config feko}
				\end{figure}

				\subsection{Spectral efficiency and energy efficiency}
			In this section, we present the spectral efficiency and energy efficiency results for our proposed hybrid reconfigurable parasitic uniform planar array (HRP-UPA) architecture. 
			We also compare with the following three architectural benchmarks.
			\begin{itemize}[leftmargin=*]
				\item Fully digital uniform linear array (FD-ULA):  This array is obtained by excluding all parasitic antennas from the HRP-UPA, i.e., $N_\sfP=0$. The purpose of this comparison is to investigate the gain in spectral and energy efficiency achieved by adding parasitic reconfigurable antenna elements to an FD-ULA architecture, thus converting it to an HRP-UPA.
				\item Fully digital uniform planar array (FD-UPA): This array is obtained by using digitally controlled sources instead of the reconfigurable reactances on the parasitic elements of the HRP-UPA. The purpose of this comparison is to have a theoretical upper bound on the spectral efficiency for the planar array which is attained when all antennas are configured with digitally controlled sources.
				\item Hybrid phase-shifter sub-connected uniform planar array (HPS-UPA):  
				Each antenna is tuned by a phase shifter.  All antennas in each row (parallel to $x$ axis) are connected together to a single active source through a power-splitter.  Each row has an active digitally controlled source. So, the total number of active sources equals number of rows in the planar array.
				The purpose of this comparison is to analyze the spectral efficiency and energy efficiency tradeoff between conventional phased-array based hybrid beamforming and the proposed hybrid parasitic-array. 
			\end{itemize}

		\begin{figure}
			\centering
			\includegraphics[width=0.45\textwidth]{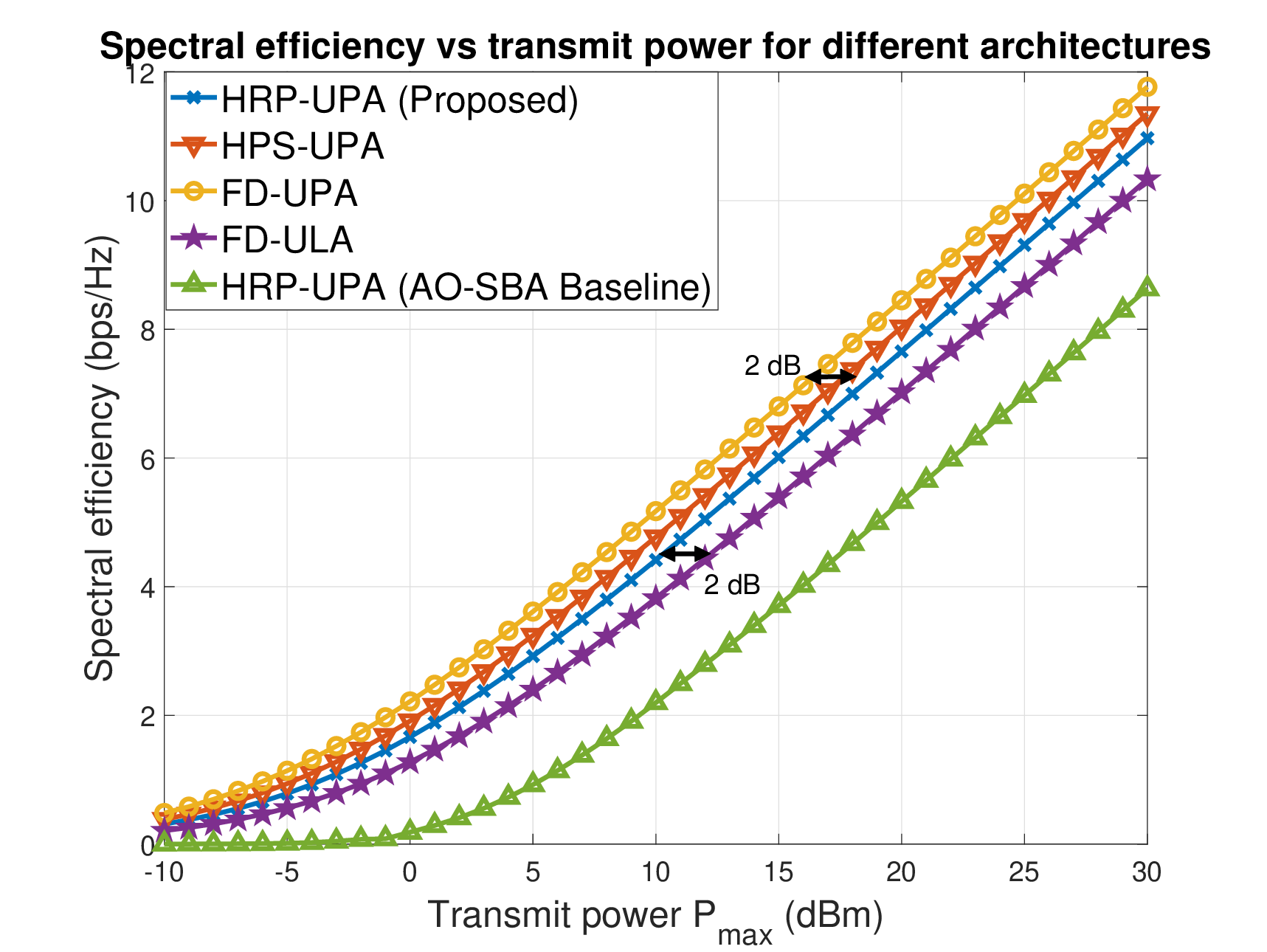}    
			\caption
			{The spectral efficiency is shown as a function of the transmit power $P_{\mathsf{max}}$  for the proposed method and different benchmarks. 
				We see that for our proposed architecture, a 2 dB transmit power reduction is possible compared to 	an FD-ULA with 6 active elements by adding two parasitic elements per active antenna, thus converting it to an HRP-UPA.
			}
			\label{fig:SE_vs_Pmax_NA6_NP2}
		\end{figure}
		
		For a given $\mathsf{SNR}$, the spectral efficiency (in bps/Hz) is defined as $\mathsf{SE}=\log_2(1+\mathsf{SNR})$.
		The  $\mathsf{SNR}$ for HRP-UPA, FD-ULA,  
		FD-UPA, and HPS-UPA  are computed using the same angle of departure and path gains for the wireless channel and fixed value of total radiated transmit power $P_{\mathsf{max}}$.
		For  HRP-UPA, we use the proposed solution for the parasitic reactance and active antenna current from \eqref{eqn: Izr star closed-form} and \eqref{eqn: iA star}. For FD-ULA and
		FD-UPA, the optimal active antenna current solution is used from our prior work~\cite{deshpande2023analysis}. For HPS-UPA, we configure the phase shifters such that they apply the conjugate phase of the channel coefficient.
		For the sub-connected array, we use the precoding matrix $\bF_{\mathsf{PS}}\in \bbC^{ N_\sfA \times N_\sfA (N_\sfP+1)}$ where each row has $(N_\sfP+1)$ non-zero complex phases.  We define an effective channel for  HPS-UPA similar to \eqref{eqn: heff def} as $\bh_{\mathsf{eff, PS}}=\bF_{\mathsf{PS}}\bh $ and an effective impedance matrix similar to \eqref{eqn: Z eff def} as $\bm{\sfZ}_{\mathsf{eff,PS}}=\bF_{\mathsf{PS}}^\rmc \cR\{\bm{\sfZ}_{\mathsf{TX}}\} \bF_{\mathsf{PS}}^\rmT$.
		The active currents for HPS-UPA  are set similar to \eqref{eqn: iA star} using $\bh_{\mathsf{eff, PS}}$ and  $\bm{\sfZ}_{\mathsf{eff,PS}}$.
		 The $\mathsf{SNR}$ expressions for these  different architectures are summarized in Table~\ref{tab: SNR and power consumption comparison}.
		 We also compare the proposed closed-form solution for the parasitic reactance and active antenna current to a baseline approach called alternating optimization stochastic beamforming algorithm (AO-SBA) from \cite{papageorgiou2018efficient}. The baseline is based on minimizing a metric dependent on the cosine similarity between the desired  and the actual current vector.
		 
		 \begin{figure}
		 	\centering
		 	\includegraphics[width=0.45\textwidth]{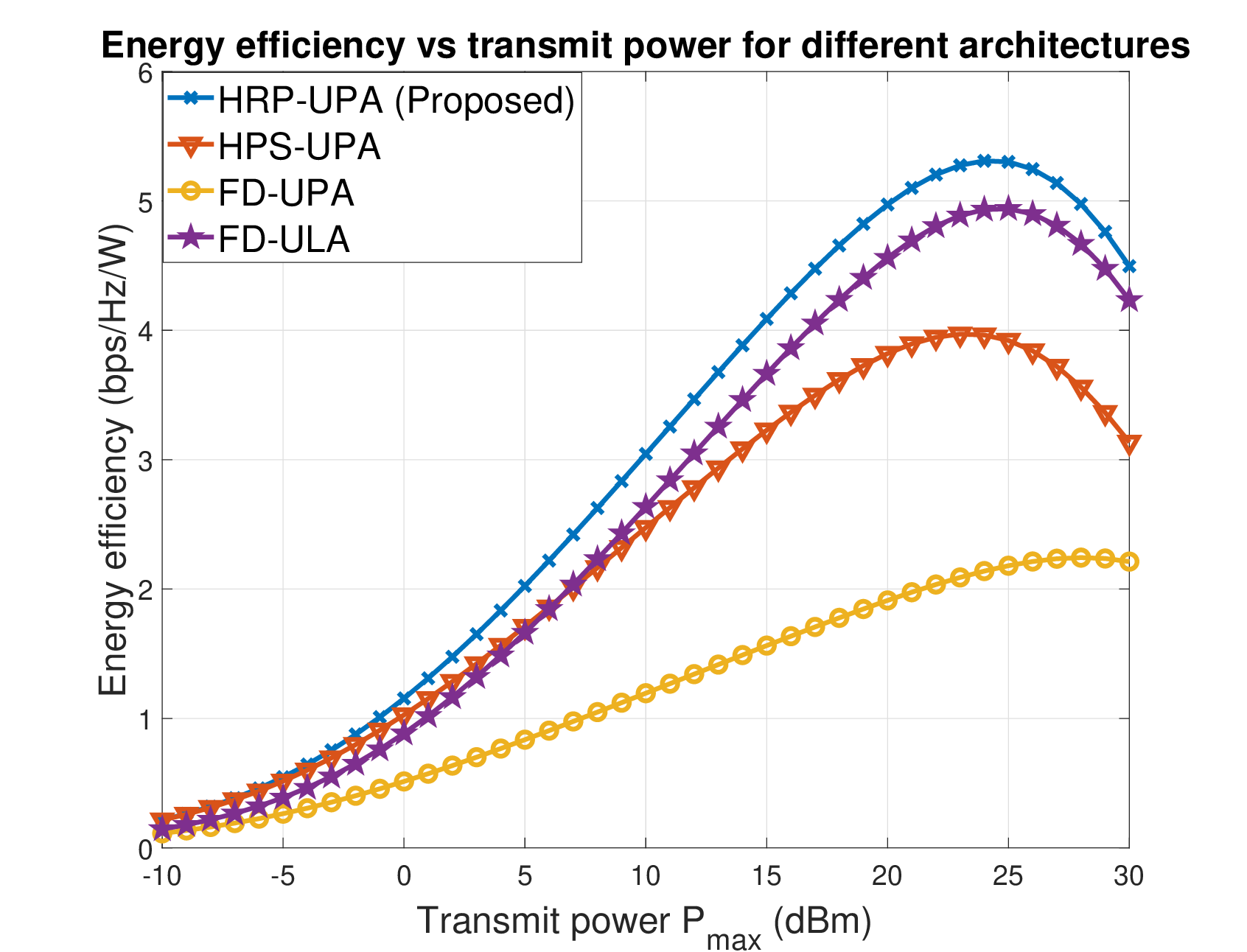}    
		 	\caption
		 	{The energy efficiency is shown as a function of the transmit power $P_{\mathsf{max}}$ for the proposed method and different benchmarks. 
		 		Our proposed HRP-UPA outperforms all benchmarks in terms of energy efficiency. The transmit power at which the energy efficiency peaks is also the lowest for our proposed architecture.
		 	}
		 	\label{fig:EE_vs_Pmax_NA6_NP2}
		 \end{figure}
		 
		 For the simulations, we set the distance between transmit array and receive antenna as $r=250$ m, bandwidth $\mathsf{BW}=20$ MHz, antenna noise temperature $T_{\sfA}=300$ K, Boltzmann's constant $k_{B}=1.38\times 10^{-23}$ J/K, and radiation resistance $R_r=95.5$~$\Omega$. These values are used to compute the constant $\frac{\gamma^2}{\sigma^2}= \left(\frac{\lambda_\sfc}{4\pi r}\right)^2\frac{R_r}{4 k_B T_{\sfA} \mathsf{BW}}$.  In our simulations, we vary the value of $P_{\mathsf{max}}$ from -10 dBm to 30 dBm.
		Assuming total power consumption is $P_\mathsf{total}$, the energy efficiency (in bps/Hz/W) is defined as $\mathsf{EE}=\frac{\mathsf{SE}}{P_\mathsf{total}}$.
		 We use the power consumption values from the conventional hybrid beamforming literature~\cite{7370753}.
		 For FD-ULA and FD-UPA, the number of RF chain equals the number of antennas. We assume that each RF chain which consists of data converter, mixer, local oscillator, filter, and amplifier, consumes a power of $P_{\mathsf{RFC}}=240 $~mW. For HPS-UPA, there are $N_\sfA$ active RF chains and each antenna is connected to a phase-shifter which consumes a power of $P_{\mathsf{PS}}=30$ mW.
		 A 1-to-3 splitter with insertion loss of 1.6 dB is used to connect the RF chain to the phase-shifters of each antenna in the sub-panel~\cite{10179153}.  Each phase-shifter has 0.7 dB insertion loss~\cite{10179153}. The total insertion loss in HPS-UPA is 2.3 dB. We multiply the factor $\epsilon = 10^{2.3/10}$ to $P_{\mathsf{max}}$ to account for the insertion losses.
		  We assume that the power required to control the reconfigurable reactance through varactor diodes is negligible, i.e., $P_\mathsf{VAR}\approx 0$ similar to \cite{10584442}.
		  For all architectures, we summarize the power consumption expressions in Table~\ref{tab: SNR and power consumption comparison}.

		\begin{table}
			\centering
			\begin{tabular}{ |c|c|c| } 
				\hline
				Architecture & SNR &  $P_\mathsf{total}$\\  \hline
				HRP-UPA  & $\frac{P_{\mathsf{max}}\gamma^2}{\sigma^2}(\bh_{\mathsf{eff}}^{\star})^{\ast}(\bm{\sfZ}_{\mathsf{eff}}^{\star})^{-1}\bh_{\mathsf{eff}}^{\star}$  & \thead{$P_{\mathsf{max}}+ N_\sfA P_{\mathsf{RFC}}$\\ $+  N_\sfA N_\sfP P_\mathsf{VAR} $}\\ \hline
				FD-ULA   &$\frac{P_{\mathsf{max}}\gamma^2}{\sigma^2}\bh_{\sfA}^{\ast}(\cR\{\bm{\sfZ}_{\sfA}\})^{-1}\bh_{\sfA}$  &  $P_{\mathsf{max}}+N_\sfA P_{\mathsf{RFC}}$\\ \hline
				FD-UPA  &  $\frac{P_{\mathsf{max}}\gamma^2}{\sigma^2}\bh^{\ast}(\cR\{\bm{\sfZ}_{\mathsf{TX}}\})^{-1}\bh$ & \thead{$P_{\mathsf{max}}$\\$+N_\sfA(N_\sfP+1) P_{\mathsf{RFC}}$  }  \\ \hline
				HPS-UPA  &  $\frac{P_{\mathsf{max}}\gamma^2}{\sigma^2}\bh_{\mathsf{eff, PS}}^{\ast}(\bm{\sfZ}_{\mathsf{eff,PS}})^{-1}\bh_{\mathsf{eff, PS}}$ &  \thead{$\epsilon P_{\mathsf{max}}+ N_\sfA P_{\mathsf{RFC}}$ \\$+ N_\sfA(N_\sfP+1) P_{\mathsf{PS}}$} \\ \hline
			\end{tabular}
			\caption{Comparing SNR and $P_\mathsf{total}$ of different architectures. }
			\label{tab: SNR and power consumption comparison}
		\end{table}

			We generate simulation results which are averaged over 5000 realizations of the wireless channel $\bh$ defined in Section~\ref{sec: multiport ckt model chan mod gen} with $L=4$, $\theta_{\ell}\sim \cU [-\pi, \pi]$, and $\alpha_\ell \sim \mathcal{N}_\bbC(0,1)$.
			We plot the spectral efficiency in Fig.~\ref{fig:SE_vs_Pmax_NA6_NP2} and energy efficiency in Fig.~\ref{fig:EE_vs_Pmax_NA6_NP2} as a function of the transmit power  $P_{\mathsf{max}}$.
			We observe from Fig.~\ref{fig:SE_vs_Pmax_NA6_NP2} that the proposed approach for configuring the HRP-UPA outperforms the AO-SBA baseline by a huge margin. The AO-SBA is sub-optimal because it is optimized for aligning the beam patterns using cosine similarity and disregarding the coupled magnitude and phase constraint. Our proposed approach performs drastically better because we use the beamforming gain objective subject to explicit phase and magnitude constraint due to the parasitic reactance.
			The HRP-UPA with 6 active antennas and 12 parasitic antennas outperforms the FD-ULA with 6 active antennas in terms of spectral efficiency.  The 12 parasitic antennas provide additional reconfigurability, improving spectral efficiency without any additional power overhead compared to the FD-ULA.   The spectral efficiency  performance of HRP-UPA is slightly lower than HPS-UPA because of the coupled magnitude and phase constraint in parasitic beamforming weight unlike the unit modulus phase shifts.
			Compared to the FD-UPA, which has the highest spectral efficiency, the HRP-UPA requires an additional 2 dB of transmit power to achieve the same spectral efficiency. Despite needing more transmit power than FD-UPA and HPS-UPA for the same spectral efficiency, the proposed HRP-UPA has the highest energy efficiency. 
			Our results demonstrate that using parasitic reconfigurable elements in a hybrid architecture is beneficial from an energy efficiency perspective.

		\begin{figure}
			\centering
			\includegraphics[width=0.45\textwidth]{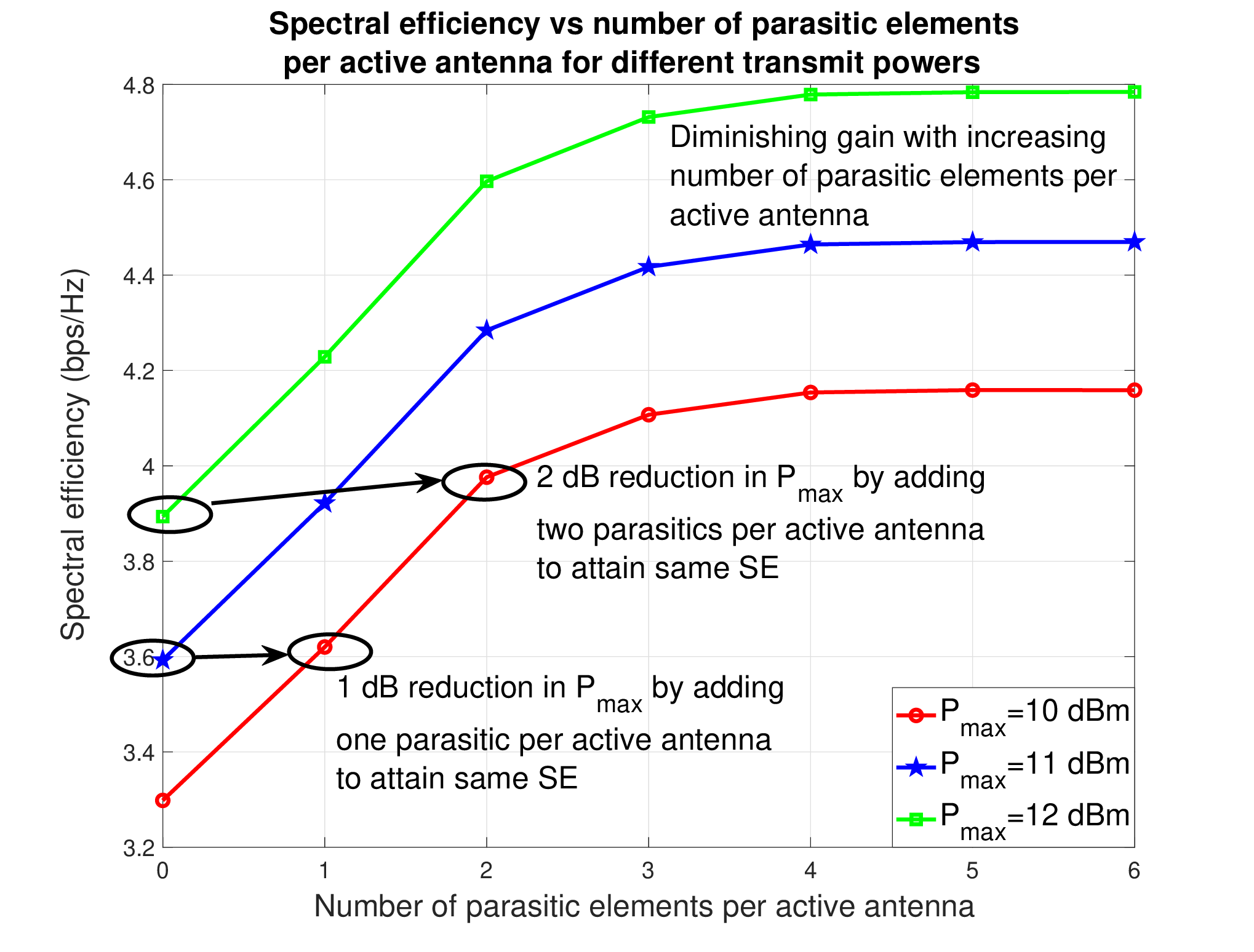}    
			\caption
			{The spectral efficiency is shown as a function of the number of parasitic elements per active antenna.
				We observe that  it is sufficient to have just two parasitic elements corresponding to each active antenna as diminishing gains are observed beyond $N_\sfP=2$. Moreover, we observe a 1 dB power reduction for $N_\sfP=1$ and 2 dB power reduction for $N_\sfP=2$ compared to an FD-ULA with $N_\sfP=0$.
			}
			\label{fig:SE_vs_Pmax_NA4_NP1234}
		\end{figure}

			\subsection{Variation in the number of  parasitic elements, active antennas, and parasitic element spacing}
			
			 We now vary the number of parasitic elements corresponding to each active antenna. 
			We generate the planar array  for $N_\sfA=4$ and vary $N_\sfP$ from 0 to 6.
		In Fig.~\ref{fig:SE_vs_Pmax_NA4_NP1234}, we plot the spectral efficiency  as a function of $N_\sfP$. 
We observe that the spectral efficiency increases up to $N_\sfP=2$ and then saturates.
			This shows that increasing the number of parasitic elements beyond two per active antenna leads to diminishing gains.  This is because beyond $N_\sfP=2$, the distance between parasitic elements and active antenna increases which leads to lower mutual coupling and reduces the beamforming capability.  Hence, placing fewer  parasitic elements closer to the active antenna is sufficient to attain reasonable spectral efficiency improvement.

			
		We now vary the number of active antennas $N_\sfA$   in the $y$ direction. 
			For this simulation, we use the impedance matrices for the planar dipole array using theoretical closed-form expressions similar to that used in Section~\ref{subsec: Numerical results for LOS beamforming gain optimization} and our prior work~\cite{deshpande2023analysis}. We fix $P_{\mathsf{max}}=10$ dBm and  $N_\sfP=2$.
		In Fig.~\ref{fig:SE_vs_NA}, we plot the spectral efficiency for the proposed hybrid parasitic reconfigurable array and compare it with the three benchmarks as a function of $N_\sfA$. For all architectures, the spectral efficiency increases with $N_\sfA$. We observe that in comparison with FD-ULA and FD-UPA, the number of active antennas required reduce by a factor of $\frac{2}{3}$ in an HRP-UPA.  Hence, by including two parasitic elements for each active antenna, the system is able to enhance the spectral efficiency which allows reduction in the number of active antennas to attain the same spectral efficiency performance. This observation can be useful to improve the spectral efficiency of   power constrained devices with limited number of active RF chains by leveraging the additional reconfigurability offered by parasitic elements.
		To conclude, Fig.~\ref{fig:SE_vs_Pmax_NA4_NP1234} and Fig.~\ref{fig:SE_vs_NA} show that parasitic elements enable a  low-power and lost-cost array architecture because of the reduction in transmit power and number of RF chains.
		
		The mutual coupling plays a significant role in the performance of the parasitic array. The  mutual coupling increases with decrease in spacing but not monotonically~\cite{balanis2015antenna}. Its effect on the beamforming gain and spectral efficiency is also not monotonic. It was shown in \cite{ivrlac_toward_2010} that the transmit beamforming gain trend with spacing varies with the angle. Specifically, the highest gain in the endfire direction occurs when inter-element spacing tends to 0, and for   broadside  when it is  between 0.5 and 1. 
		The spectral efficiency metric captures the average effect over different multi-path channel realizations. This observation motivates us to simulate the  spectral efficiency trend of a parasitic array by varying the normalized inter-element spacing along the $x$ direction, $\frac{d_\sfx}{\lambda_\sfc}$.
	From Fig.~\ref{fig:SE_vs_dx},	we observe that the highest spectral efficiency is observed for moderate inter-element spacing, i.e., $0.15 <\frac{d_\sfx}{\lambda_\sfc} < 0.4$. The mutual coupling within this range is sufficient for good spectral efficiency performance of a parasitic array.

		\begin{figure}
			\centering
			\includegraphics[width=0.45\textwidth]{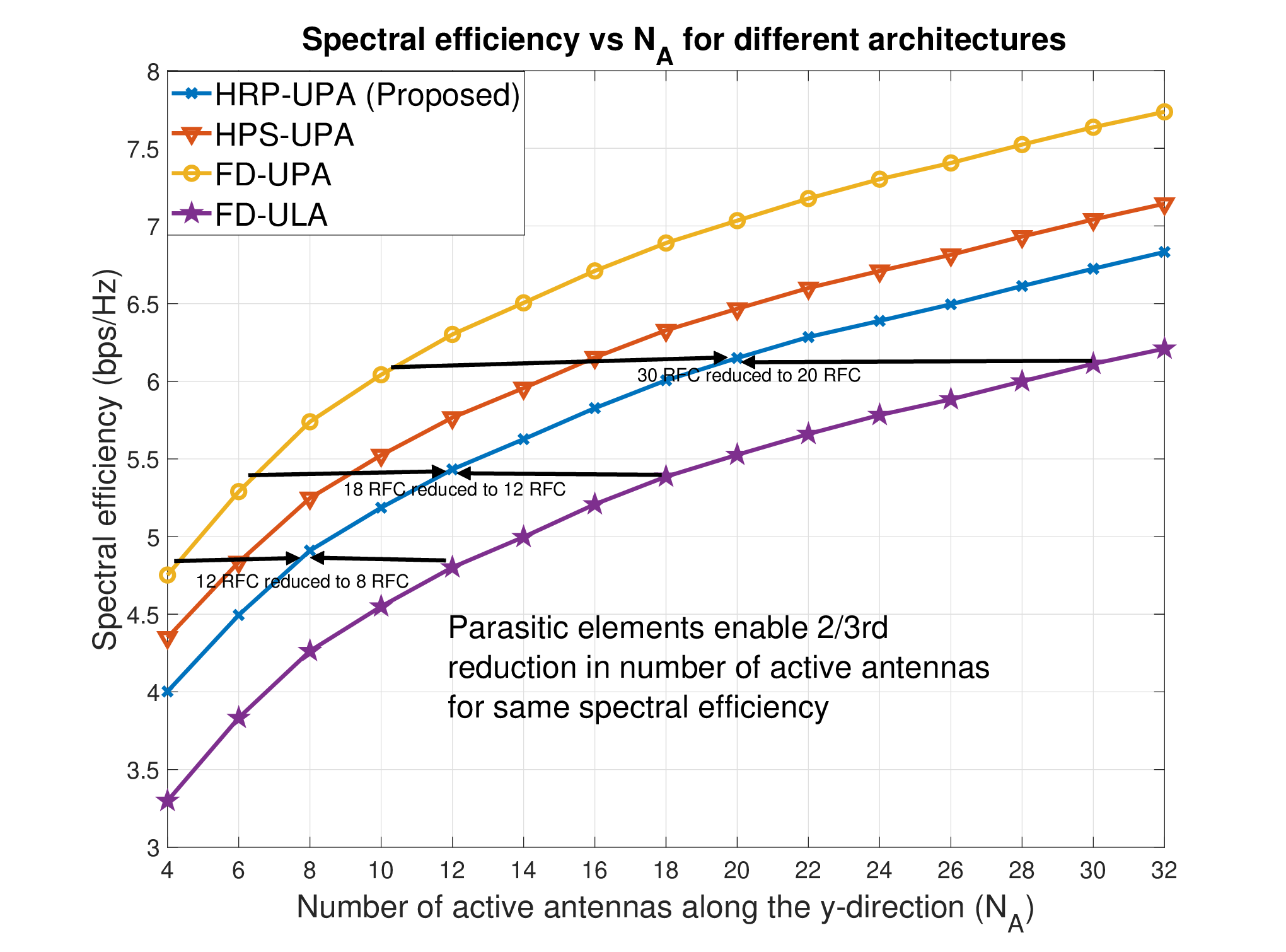}    
			\caption
			{The spectral efficiency is shown as a function of $N_\sfA$.
			 The spectral efficiency of our proposed HRP-UPA with $N_\sfP=2$ and $N_\sfA= 8/ 12/ 20$ is same as the  FD-ULA with $N_\sfA=12/ 18/ 30$ antennas and FD-UPA with  $N_\sfA( N_\sfP+1)=12/ 18/ 30$ antennas respectively. 	The reconfigurability of parasitic elements is leveraged to reduce the number of active antennas. This enables cost reduction in the RF hardware.
			}
			\label{fig:SE_vs_NA}
		\end{figure}

			\begin{figure}
			\centering
			\includegraphics[width=0.5\textwidth]{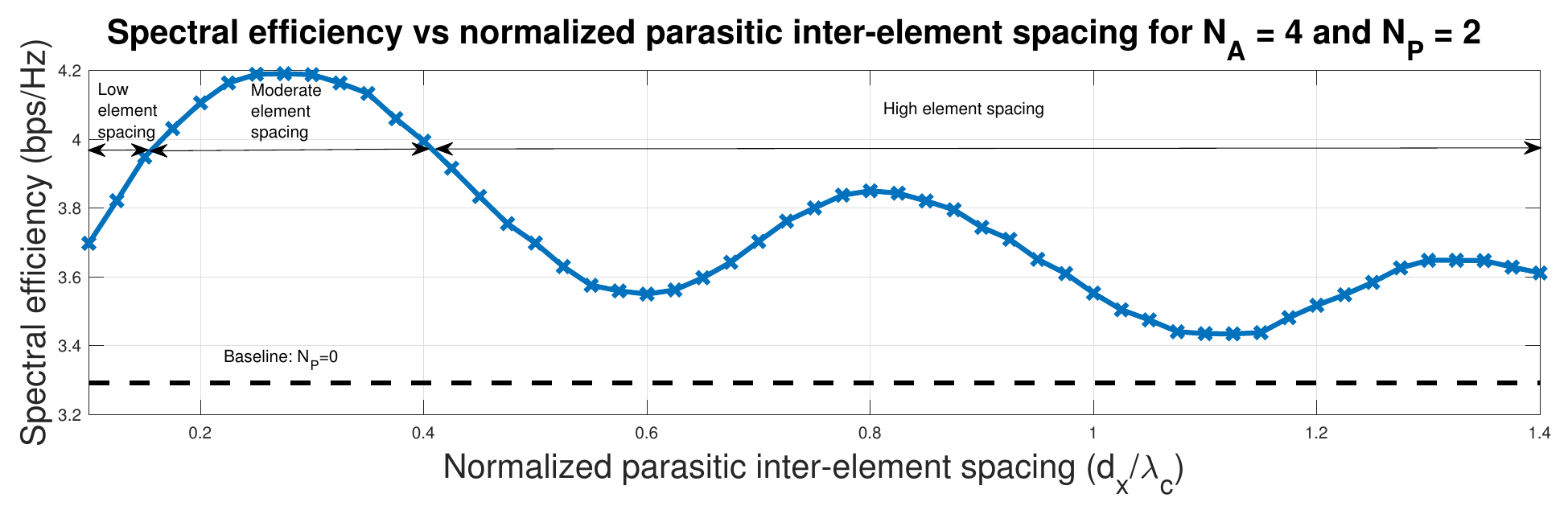}    
			\caption
			{The spectral efficiency is shown as a function of the normalized parasitic inter-element spacing. We observe that for a parasitic array with moderate inter-element spacing, the spectral efficiency is the highest. 
			}
			\label{fig:SE_vs_dx}
		\end{figure}

				\section{Conclusions and future work}\label{sec: Conclusion and future work }

			In this paper, we applied a circuit theory modeling approach to a hybrid reconfigurable parasitic array with multiple active antennas and several parasitic elements controlled through tunable reactance. 
			As validated through the Feko software, the circuit theory approach is useful because we can use the proposed model for different parasitic configurations instead of doing repetitive simulations.
		From the closed-form solution of the optimal reconfigurable parasitic reactance  for a LOS channel, we identify a fixed reactance term common to all antennas and a dynamic term depending on the target angle.
For the general case of the hybrid array, we leverage this closed-form solution for parasitic reactance. We  configure the active antenna currents to satisfy the transmitted power constraint based on the parasitic reactance configuration.
Using this closed-form solution, we simulate the spectral efficiency and energy efficiency of the hybrid architecture for different array dimensions and element spacings.
			  Compared to conventional benchmarks like fully digital array and hybrid sub-connected phase shifter architectures, the proposed parasitic reconfigurable array is more energy-efficient. The same spectral efficiency performance is achievable with a reduced number of RF chains and  transmit power by adding reconfigurable parasitic elements to a fully digital array. We also identified a moderate inter-element spacing region for the parasitic array design to provide higher spectral efficiency.
			
			In future work, we propose extending the MISO parasitic array to the MIMO case with reconfigurable parasitic antennas on both the transmitter and receiver. 
			 Analyzing reconfigurable parasitic arrays with realistic matching network constraints is an important future direction. Another open challenge is a  wideband generalization  with frequency-selective models for parasitic antennas and matching networks~\cite{10536063}.

				\appendices

				\section{Proof of $\cG(\theta, \bm{\sfZ}_{\mathsf{R}}) \rightarrow \widehat{\cG}(\theta, \bW) $ for small $\|  \bm{\sfE}_\sfP\|_\rmF$}\label{app: proof of approximation}
			
				Let $ \Psi_1= 1  - \ba^\rmT_{\sfP}(\theta)\bW\bm{\sfz}_{\sfm} $  and $\Psi_2= \ba^\rmT_{\sfP}(\theta) \bW \bm{\sfE}_\sfP \bW \bm{\sfz}_{\sfm} $. We have $\widetilde{\cG}(\theta, \bm{\sfZ}_{\mathsf{R}}) = \cG(\theta, \bm{\sfZ}_{\mathsf{R}})- \widehat{\cG}(\theta, \bW) =  |\Psi_2|^2 + 2 \cR\{\Psi_1^\rmc \Psi_2\} - \cO(\| \bm{\sfE}_\sfP \|^2_\rmF) $. Using matrix norm inequality, we  have $|\Psi_2|< \|\ba_{\sfP}(\theta) \|\|\bW \|^2_\rmF \|\bm{\sfz}_{\sfm}\| \| \bm{\sfE}_\sfP\|_\rmF  = \kappa \| \bm{\sfE}_\sfP\|_\rmF $ where $\kappa $ is a constant depending on $\|\ba_{\sfP}(\theta) \|$, $\|\bW \|_\rmF$, and $\|\bm{\sfz}_{\sfm}\|$. As  $\|  \bm{\sfE}_\sfP\|_\rmF \rightarrow 0$, we get $\widetilde{\cG}(\theta, \bm{\sfZ}_{\mathsf{R}})\rightarrow 0$, i.e., $ \cG(\theta, \bm{\sfZ}_{\mathsf{R}}) \rightarrow \widehat{\cG}(\theta, \bW) $.

				
				\section{Proof of Theorem~\ref{thm: closed-form solution for phase}}\label{proof: thm: closed-form solution for phase}
				Let $\chi_1= 1  - \frac{\zeta}{2}\ba^\rmT_{\sfP}(\theta)\bm{\sfz}_{\sfm}$ and $\chi_2=-\frac{\zeta}{2}\ba^\rmT_{\sfP}(\theta)\boldsymbol{\Phi}\bm{\sfz}_{\sfm}$.
				We expand the objective in \eqref{eqn:  Ghat in terms of Phi} as $	\widehat{\cG}(\theta, \boldsymbol{\Phi}) = |\chi_1+\chi_2|^2$.
				The term $\chi_1$ is a constant complex number which does not depend on the beamforming matrix $\boldsymbol{\Phi}$. The term $\chi_2$ depends on $\boldsymbol{\Phi}$.
				Let $\chi_{2,i}=-\frac{\zeta}{2} e^{\sfj \phi_i} [\ba_{\sfP}(\theta)]_i  [\bm{\sfz}_{\sfm}]_i$.
				 We express $\chi_2=  \sum_{i=1}^{N_\sfP} \chi_{2,i}.$
				Using this expression, we expand $\widehat{\cG}(\theta, \boldsymbol{\Phi})$ as
				$	\widehat{\cG}(\theta, \boldsymbol{\Phi}) = |\chi_1|^2+2\cR\left\{  \chi_1^{c} \sum_{i=1}^{N_\sfP} \chi_{2,i} \right\}+   \left|\sum_{i=1}^{N_\sfP} \chi_{2,i}\right|^2.$
				The second term in this expansion  is maximized when the argument of the complex number $\chi_1^{c} \sum_{i=1}^{N_\sfP} \chi_{2,i}$ is zero. This is obtained by setting $\angle(\chi_{2,i})=2m\pi + \angle(\chi_1)$ for any integer $m$.  This setting also maximizes the third term  in the expansion as all phases are aligned.  Simplifying, we get the  result in \eqref{eqn:  Phi i optimal}.

						\bibliographystyle{IEEEtran}
					\bibliography{references_Nitish.bib}

					\end{document}